\documentclass{conm-p-l}
\usepackage{amsmath}
\usepackage{amsthm}
\usepackage{amscd}
\usepackage{amssymb}
\usepackage[T1]{fontenc}
\usepackage{url}
\usepackage[inline,shortlabels]{enumitem}

\newtheoremstyle{standard}
  {}
  {}
  {\rmfamily\mdseries\itshape}
  {\parindent}
  {\rmfamily\mdseries\scshape}
  {.}
  {.7em}
  {\thmnumber{#2. }\thmname{#1}\thmnote{ \textmd{(#3)}}
}

\newtheoremstyle{alt}
  {}
  {}
  {\rmfamily\mdseries\upshape}
  {\parindent}
  {\rmfamily\mdseries\scshape}
  {.}
  {.7em}
  {\thmnumber{#2. }\thmname{#1}\thmnote{ \textmd{(#3)}}
}

\newcounter{entry}[section]

\renewcommand{\theentry}{\thesection.\arabic{entry}}
\newcommand{\entry}[1][\hspace{-.5em}]{%
 \vspace{.5\baselineskip}\par%
 \refstepcounter{entry}%
 {\rmfamily\mdseries\upshape\theentry.\hspace{.5em}\itshape #1\hspace{.5em}---\hspace{.6em}}%
}

\theoremstyle{standard}

\newtheorem{lemma}[entry]{Lemma}
\newtheorem{proposition}[entry]{Proposition}
\newtheorem{corollary}[entry]{Corollary}
\newtheorem*{theorem*}{Theorem}
\newtheorem*{lemma*}{Lemma}
\newtheorem*{proposition*}{Proposition}
\newtheorem*{corollary*}{Corollary}

\theoremstyle{alt}
\newtheorem{definition}[entry]{Definition}
\newtheorem{example}[entry]{Example}
\newtheorem{remark}[entry]{Remark}
\newtheorem{problem}[entry]{Problem}
\newtheorem*{definition*}{Definition}
\newtheorem*{example*}{Example}
\newtheorem*{remark*}{Remark}
\newtheorem*{problem*}{Problem}

\numberwithin{equation}{entry}

\newcommand{\beq}{\begin{equation}}
\newcommand{\eeq}{\end{equation}}
\newcommand{\beqv}{\begin{equation*}}
\newcommand{\eeqv}{\end{equation*}}

\DeclareMathOperator{\im}{im}

\DeclareMathOperator{\Spec}{Spec}
\DeclareMathOperator{\Proj}{Proj}
\DeclareMathOperator{\Sec}{Sec}

\DeclareMathOperator{\res}{res}
\DeclareMathOperator{\ev}{ev}
\DeclareMathOperator{\id}{id}
\DeclareMathOperator{\codim}{codim}

\newcommand{\F}{\mathbb{F}}

\newcommand{\GG}{\mathbf{G}}
\newcommand{\FF}{\mathbf{F}}

\newcommand{\PP}{\mathbf{P}}
\newcommand{\fI}{\mathfrak{I}}

\newcommand{\cA}{\mathcal{A}}
\newcommand{\cB}{\mathcal{B}}
\newcommand{\cC}{\mathcal{C}}
\newcommand{\cE}{\mathcal{E}}
\newcommand{\cK}{\mathcal{K}}
\newcommand{\cL}{\mathcal{L}}
\newcommand{\cO}{\mathcal{O}}
\newcommand{\cP}{\mathcal{P}}
\newcommand{\cQ}{\mathcal{Q}}
\newcommand{\cR}{\mathcal{R}}

\newcommand{\longto}{\longrightarrow}
\newcommand{\tens}{\otimes}
\newcommand{\moins}{\setminus}

\newcommand{\deux}[1][2]{^{\langle #1\rangle}}
\newcommand{\abs}[1]{\lvert #1\rvert}

\newcommand{\linspan}[1]{\langle #1\rangle}
\newcommand{\inj}{\hookrightarrow}
\newcommand{\surj}{\twoheadrightarrow}

\let\phi\varphi
\let\epsilon\varepsilon
\let\subset\subseteq
\let\supset\supseteq
\let\subsetneq\varsubsetneq

\title[Quadratic hull and multiplication algorithms]{The quadratic hull of a code and\\ the geometric view on multiplication algorithms}
\author{Hugues Randriambololona}
\date{}

\begin{document}

\setlist{leftmargin=2\parindent}

\maketitle

\begin{abstract}
We introduce the notion of quadratic hull of a linear code, and give some of its properties.
We then show that any symmetric bilinear multiplication algorithm for a finite-dimensional algebra over a field
can be obtained by evaluation-interpolation at simple points (i.e. of degree and multiplicity $1$) on a naturally
associated space, namely the quadratic hull of the corresponding code. This also provides a geometric answer to some questions
such as: which linear maps actually are multiplication algorithms, or which codes come from supercodes (as asked in~\cite{STV92}).
We illustrate this with examples, in particular we describe the quadratic hull of all the optimal algorithms computed
in~\cite{BDEZ12} for small algebras.

In our presentation we actually work with multiplication reductions. This is a generalization of multiplication algorithms,
that allows for instance evalu\-ation-interpolation at points of higher degree and/or with multiplicities, and also includes
the recently introduced notion of ``reverse multiplication-friendly embedding'' from~\cite{CCXY18}. All our results hold
in this more general context.
\end{abstract}


\section{Introduction}

\entry
Early remarks of Goppa \cite{Goppa83} and Lachaud \cite[\S5.10]{Lachaud86},
made more precise by Pellikaan, Shen, and van Wee in~\cite{PSW91},
show that any linear code is an evaluation code on an algebraic curve.
In some sense this work presents an analogue for (symmetric bilinear) multiplication algorithms:
any such multiplication algorithm is an evaluation-interpolation algorithm on
some algebraic space.

\entry
The theory of bilinear complexity started with the celebrated algorithms
of Karatsuba~\cite{KaOf63}, that allows to multiply two $2$-digit numbers with $3$ elementary
multiplications instead of $4$, and of Strassen~\cite{Strassen69}, that allows to multiply two $2\times2$
matrices with $7$ field multiplications instead of $8$.
Used recursively, these algorithms then allow to multiply numbers with a large number of digits,
or matrices of large size,
with a dramatic improvement on complexity over the naive methods.

General considerations on bilinear algorithms, in particular in relation with tensor decompositions,
can be found in \cite{Strassen73}\cite{BrDo78}.
In \cite{FiZa77}\cite{Winograd77} the theory is developed in terms
of multiplicative complexity of algebras,
with a special focus on quotient algebras of a polynomial ring in one indeterminate,
especially over a finite field.

A bilinear algorithm of length $n$ for a bilinear map $B$ over a field $\FF$
reduces the computation of $B$ to $n$ multiplications in $\FF$, plus some
fixed linear operations.
Expressed in coordinates, the bilinear map $B$ gives rise to a collection
of bilinear forms $B_1,\dots,B_k$.
The goal is then to find $n$ products of two linear forms, i.e. $n$ bilinear
forms of rank~$1$, whose linear span contains $B_1,\dots,B_k$.
This is how the problem is usually formulated in these early works.

In this text the bilinear map $B$ will always be the multiplication map
in a finite dimensional commutative algebra $\cA$.
Also we will consider only algorithms whose structure reflects this commutativity.
Rephrasing the definition in this particular context and in a more
coordinate-free way, a \emph{symmetric bilinear multiplication algorithm}
of length $n$ for $\cA$ over $\FF$ is a pair
of linear maps $\phi:\cA\to\FF^n$ and $\omega:\FF^n\to\cA$, such that
for all $a,a'\in\cA$ we have
\beqv
aa'=\omega(\phi(a)*\phi(a'))
\eeqv
where $*$ denotes componentwise multiplication in $\FF^n$.

\entry
In their important paper \cite{ChCh88}, Chudnovsky and Chudnovsky
highlighted several links between multiplication algorithms and codes.
Then, using a construction similar to Goppa's algebraic geometry codes,
they showed that, at least over a base finite field of (not too small)
square order,
\emph{evaluation-interpolation} on curves allows to produce
multiplication algorithms in an extension field of arbitrary
large degree, whose complexity remains linear with the degree.

Several improvements were then proposed by Shparlinski, Tsfasman, and
Vladut in \cite{STV92}: 
they showed how to deduce linearity of the complexity over an arbitrary
base finite field;
they introduced the notion of supercode,
that captures the part of the procedure that is mere linear algebra,
and they reduced what remains to a cleanly posed problem in algebraic geometry;
and they corrected certain statements or made some arguments in the proofs
of \cite{ChCh88} more precise,
in particular concerning the choice of the curves in the construction.
Then in \cite{Ballet99}, Ballet replaced certain geometric conditions with
more manageable numerical criteria.

In \cite{BaRo04}, Ballet and Rolland extended the Chudnovsky-Chudnovsky
method by allowing evaluation-interpolation at points of higher degree,
while \cite{Arnaud} used points with multiplicities, and \cite{CeOz10}
a combination of both.

Further contributions to the method were then proposed in \cite{HR-ChCh+}:
it improved the use of points of higher degree and with multiplicities on curves
(as a particular case of evaluation-interpolation
at arbitrary closed subschemes on arbitrary schemes);
it distinguished symmetric and asymmetric evaluation-interpolation algorithms;
and it fixed a geometric construction of \cite{STV92}, that allows to improve
the asymptotic bound of \cite{ChCh88}.

In another direction, works such as \cite{Oseledets08}\cite{BDEZ12}\cite{Covanov19}
aimed at exhaustively finding multiplication algorithms of optimal length
for algebras of small cardinality.

For recent results on bilinear multiplication algorithms in extensions of finite fields we refer to the extensive survey \cite{BCPRRR19}.
For a more general introduction to bilinear complexity we refer to \cite[Ch.~14]{BCS}.

\entry
Multiplication algorithms are closely related to multiplicative
(or ``arithmetic'') linear secret sharing schemes.
We refer to \cite{CDM00}\cite{ChCr06}
for foundational works on this topic, and to the book \cite{CDN}
for more on secure multi-party computation and secret sharing.
Roughly speaking, a multiplicative linear secret sharing scheme is
a multiplication algorithm that further admits certain threshold properties.

Because of this link, plain multiplication algorithms (without these threshold properties)
can be used as a technical tool in the construction of multiplicative
secret sharing schemes \cite{CCCX09}:
in this context, a multiplication algorithm that reduces a multiplication in $\F_{q^k}$ to $n$ multiplications in $\F_q$
is also called a ``multiplication-friendly embedding'' of $\F_{q^k}$ into $\F_q^n$.
In~\cite{CCXY18} the notion of ``reverse multiplication-friendly embedding'' (RMFE) is also introduced, which goes in the opposite direction:
a RMFE from $\F_q^k$ to $\F_{q^n}$ reduces $k$ multiplications in $\F_q$
to one multiplication in $\F_{q^n}$.
It turns out such RMFE can be constructed by evaluation-interpolation,
in a way that duplicates the Chudnovsky-Chudnovsky method.

This situation made it desirable to introduce a more general notion
of \emph{multiplication reduction} (Definition~\ref{def-algo}),
from an arbitrary finite-dimensional commutative algebra $\cA$
to another such arbitrary algebra $\cB$,
and to which evaluation-interpolation constructions can be applied
in a uniform framework (Proposition~\ref{evaluation-interpolation}).
As a by-product, this new notion of multiplication reduction also
allows to reinterpret the evaluation-interpolation algorithms
at points of higher degree and with multiplicities of \cite{BaRo04}\cite{Arnaud}\cite{CeOz10}\cite{HR-ChCh+}
discussed above.

For simplicity, in the remaining of this introduction we describe our
results only for multiplication algorithms.
Full statements and proofs for general multiplication reductions
will be found in the main text.
This will be best expressed in the language of abstract algebraic geometry,
although for ease of the reader we will also give alternative descriptions, in coordinates. 

\entry
Assume that $X$ is an algebraic space (e.g. a curve, a surface, etc.) together with a point $Q\in X(\cA)$
with coordinates in $\cA$, and a collection $\cP=(P_1,\dots,P_n)$ of rational points. 
In order to multiply $a,a'\in\cA$, we can proceed by evaluation-interpolation, as follows:
\begin{enumerate}[(i)]
\item Lift $a,a'$ to functions $f_a,f_{a'}$ on $X$ such that $f_a(Q)=a,f_{a'}(Q)=a'$, and then evaluate these functions at $\cP$
to get vectors $c_a=(f_a(P_1),\dots,f_a(P_n))$, $c_{a'}=(f_{a'}(P_1),\dots,f_{a'}(P_n))$.
\item Multiply $c_a$ and $c_{a'}$ componentwise, i.e. compute $y_1=f_a(P_1)f_{a'}(P_1)$, ..., $y_n=f_a(P_n)f_{a'}(P_n)$. 
\item By ``Lagrange interpolation'', find a function $h$ on $X$ that takes these values $h(P_1)=y_1$, ..., $h(P_n)=y_n$,
and then evaluate $h$ at $Q$.
Under proper hypotheses, we will get $h(Q)=aa'$ as wished.  
\end{enumerate}
This actually provides a multiplication algorithm for $\cA$:
step~(i) gives the map $\phi$,
step~(ii) gives the $n$ multiplications in $\FF$,
and step~(iii) gives $\omega$.

Our main result is that any multiplication algorithm is of this sort.
Actually we will provide two proofs for this result
(and for its generalization to multiplication reductions).
Our first version (Corollary~\ref{cortauto}) is somewhat tautological,
deriving from a formal play with the notion of supercode
(Proposition~\ref{trivialrepr}).

\entry
In a second approach, inspired by the so-called geometric view on coding theory,
we try to get a description with a more meaningful geometric content.
Let's first birefly review on which sort of spaces evaluation-interpolation is usually performed. 

Early multiplication algorithms, such as those of Karatsuba or Toom-Cook, are evaluation-interpolation algorithms on the
projective line.
Chudnovsky and Chudnovsky~\cite{ChCh88} introduced evaluation-interpolation algorithms on curves.
In \cite[Ex.~2.5]{HR-ChCh+} one can find an evaluation-interpolation algorithm
on the projective plane.

In section~\ref{dP} we will construct a multiplication algorithm for $\F_{2^5}$ over $\F_2$, of optimal length $n=13$,
by evaluation-interpolation on a del Pezzo surface.

In general we show (second half of Theorem~\ref{mainth}) that any multiplication algorithm $\phi$ can be realized as an evaluation-interpolation algorithm
on a naturally associated space, namely the \emph{quadratic hull} of the linear code image of $\phi$, defined below.

A converse result (first half of Theorem~\ref{mainth}) characterizes whether a given code is the image of a multiplication algorithm,
also in terms of its quadratic hull.
Informally, it says that any quadratic identity satisfied in the code should also be satisfied in $\cA$ (Remark~\ref{explain}).
It then translates as a geometric characterization of codes that come from a supercode (Corollary~\ref{charsupercode}).

Again, all this generalizes to multiplication reductions.

We observe a similarity between these results and those from \cite{PSW91}, in that they both are \emph{a posteriori}.
They are of no help if one's goal is to build better codes, or better multiplication
algorithms. What they provide is only a better abstract understanding of these objects.

\entry
The quadratic hull (French: \emph{enveloppe quadratique}) of a non-degenerate linear code $C\subset\FF^n$ can be seen either algebraically or geometrically.

Starting from the algebraic point of view, let $\GG\in\FF^{k\times n}$ be a generator matrix for $C$,
and let $P_1,\dots,P_n$ be the columns of $\GG$.
In \cite[sec.~3]{MMP11}, M\'arquez-Corbella, Mart\'inez-Moro, and Pellikaan set
\beqv
I_2(C)=\{q\in\FF[x_1,\dots,x_k]_2:\;q(P_1)=\dots=q(P_n)=0\},
\eeqv
the space of homogeneous quadratic forms that vanish at $P_1,\dots,P_n$,
which up to some linear transformation depends only on $C$.
Put more intrisically,
\beqv
I_2(C)=\ker(S^2C\overset{*}{\longto}C\deux),
\eeqv
where $S^2C$ is the second symmetric power of $C$,
and $C\deux$ is the \emph{square} of $C$, the linear span of pairwise $*$-products of codewords of $C$
(and similar spaces $I_t(C)$ are defined likewise for higher powers $C\deux[t]$).
They then show \cite[Th.~2]{MMP14} that if $C$ is an AG code with suitable numerical parameters, $I_2(C)$ allows to reconstruct the curve from which $C$ is defined.
Later on, these $I_t(C)$ were also considered in \cite[\S\S1.11-1.36]{HR-AGCT14}
as part of an intrisic description of the geometric view on codes and their powers.

Turning to geometry, we define the quadratic hull $Z_2(C)$ of $C$ as the zero locus of $I_2(C)$ in projective space:
\beqv
Z_2(C)=Z(I_2(C))\subset\PP^{k-1},
\eeqv
i.e. the intersection of all quadrics in $\PP^{k-1}$ that pass through $P_1,\dots,P_n$.

These two points of view, algebraic and geometric,
are equivalent, i.e. $I_2(C)$ and $Z_2(C)$ carry exactly the same information,
provided $Z_2(C)$ is considered with its scheme structure.
Indeed, giving a closed subscheme in $\PP^{k-1}$ is the same as giving its defining ideal.

When generalizing from multiplication algorithms to multiplication reductions,
codes have to be replaced with what we call $\cB$-codes.
The quadratic hull can also be defined in this context.

\entry
In general the quadratic hull of a code can be quite pathological,
and we have no a priori control on it.
It is thus interesting to have plenty of examples.

In \cite{BDEZ12} (see also \cite{Covanov19}) an algorithm is
given that allows to compute the exhaustive list of multiplication algorithms of optimal length for algebras of small cardinality.
In section~\ref{exp}, Tables 1-4, we describe the quadratic hull of the corresponding codes.
Moreover, Proposition~\ref{I2phiperpW}, Corollary~\ref{pointsZ2phi},
and Corollary~\ref{Z2phici} derive informations on the quadratic hull
directly from the linear span of the quadratic forms of rank~$1$
that constitute the multiplication algorithm. 

We complement this with several properties of the quadratic hull, of independent interest:
\begin{enumerate}[(i)]
\item The rational points $P$ of $Z_2(C)$ parameterize \emph{square-preserving} extensions of $C$, i.e. those $[n+1,k]$-codes $C_P$
with generator matrix $\GG_P=(\;\GG\;|P)$ such that
\beqv
\dim C_P\deux=\dim C\deux.
\eeqv
\item They also parameterize hyperplanes $H\subset C$ such that
\beqv
H\cdot C\subsetneq C\deux.
\eeqv
\item The rational points of the secant variety of $Z_2(C)$ parameterize a subset of
(and possibly all) the locus of hyperplanes $H\subset C$ that satisfy
\beqv
H\deux\subsetneq C\deux.
\eeqv
\end{enumerate}
The following is a slightly stronger version of the result from \cite{MMP11}\cite{MMP14} already mentioned:
\begin{enumerate}[(i)]
\setcounter{enumi}{3}
\item If $C$ is the evaluation code of some $L(D)$ at $n$ distinct rational points $P_1,\dots,P_n$
on a smooth projective curve $X$ of genus~$g$, with $\deg(D)>2g+1$ and $L(2D-\sum_iP_i)=0$, then
\beqv
Z_2(C)=X.
\eeqv
\end{enumerate}
All these properties actually extend to $\cB$-codes,
as covered in Propositions~\ref{extensions}, \ref{HC}, \ref{H2secante}, and~\ref{courbe} respectively.

Properties (iii)(iv)(v) are especially meaningful in the analysis of McEliece-type cryptosystems.
Beside \cite{MMP11}\cite{MMP14},
we refer to \cite{Wieschebrink10}\cite{MMP13}\cite{CGGOT14}\cite{MMPR14} for a sample of related works.

In (iv), it seems to be a difficult question to characterize those codes for which
the secant variety exactly coincides with the locus of hyperplanes whose square does not fill
the square of $C$. This is not true for all codes, but it holds at least
for certain MDS codes and for certain
evaluation codes on curves. Partial results of this sort will be found in~\cite{HR-secant}.

\entry[Acknowledgments]
The author greatly benefited from discussions with R.~Blache and E.~Hallouin on del Pezzo surfaces,
while preparing~\cite{BCHMNRR19}. This influenced the presentation of the example developed in section~\ref{dP}.

The author also thanks P.~Zimmermann and S.~Covanov for giving details about their results \cite{BDEZ12}\cite{Covanov19}.
In particular P.~Zimmermann provided the computer program that allowed the author to get the list of explicit optimal
formulae on which section~\ref{exp} is based.

\entry[Conventions]
In this work, by an algebra we will always mean an algebra that is commutative, associative,
with unity, nonzero, and of finite dimension (as a vector space) over a field $\FF$.

Also we will say ``multiplication algorithm (resp. reduction)'' as short for ``symmetric bilinear multiplication algorithm (resp. reduction)''.
It would not be difficult to devise asymmetric extensions of our results, but we will stick to the symmetric setting for simplicity.

If $\cB$ is an algebra and $\cP:\Spec\cB\to X$ is a morphism of schemes, then
for any invertible sheaf $\cL$ on $X$ we set $\cL|_{\cP}=\cP^*\cL$, the pullback
of $\cL$ by $\cP$. Observe that, since $\cB$ is finite, $\cL|_{\cP}\simeq\cB$
admits a trivialization.

\section{Multiplication reductions}
\label{sec:multred}

\begin{definition}
\label{def-algo}
A (symmetric bilinear) multiplication reduction
\beqv
\phi:\cA\leadsto\cB
\eeqv
from an algebra $\cA$ to an algebra $\cB$ over $\FF$,
is a linear map $\phi:\cA\to\cB$ that admits an ``adjoint'' linear
map $\omega:\cB\to\cA$ that makes the following diagram commute:\footnote{
It would make no fundamental difference to define a multiplication reduction
as the data of the pair $(\phi,\omega)$, instead of $\phi$ alone.
See \S\ref{adjoint} below for a discussion.}
\beqv
\begin{CD}
\cA\times\cA  @>m_{\cA}>>\cA\\
@V{\phi\times\phi}VV @AA{\omega}A\\
\cB\times\cB @>m_{\cB}>> \cB
\end{CD}
\eeqv
where $m_{\cA}$ and $m_{\cB}$ denote multiplication in $\cA$ and $\cB$,
respectively.

Equivalently, $\phi$ and $\omega$ should satisfy
\beqv
a\cdot a'=\omega(\phi(a)\cdot\phi(a'))
\eeqv
for all $a,a'\in\cA$.
In words, multiplication of $a$ and $a'$ in $\cA$ ``reduces'' (up to some fixed linear operations)
to multiplication of $\phi(a)$ and $\phi(a')$ in $\cB$.
\end{definition}

\begin{lemma}
A multiplication reduction $\phi:\cA\leadsto\cB$ is always injective.
\end{lemma}
\begin{proof}
Apply the formula above with $a\in\ker\phi$ and $a'=1$.
\end{proof}

\entry
Here are three important instances of multiplication reductions:
\begin{itemize}
\item A multiplication \emph{algorithm} of length $n$
for an algebra $\cA$ over $\FF$ is a multiplication reduction
\beqv
\cA\leadsto\FF^n
\eeqv
where $\cB=\FF^n$ is equipped with componentwise multiplication $*$.

Equivalently, if $\phi_1,\dots,\phi_n\in\cA^\vee$ are the components of $\phi:\cA\to\FF^n$,
and $\omega_1,\dots,\omega_n\in\cA$ are those of $\omega:\FF^n\to\cA$, one asks
\beqv
aa'=\phi_1(a)\phi_1(a')\omega_1+\cdots+\phi_n(a)\phi_n(a')\omega_n
\eeqv
for all $a,a'\in\cA$.
In words, a multiplication in $\cA$ reduces (up to some fixed linear operations)
to $n$ multiplications in $\FF$.

This also amounts to a decomposition
\beqv
T_{\cA}=\omega_1\tens\phi_1^{\tens2}+\cdots+\omega_n\tens\phi_n^{\tens2}
\eeqv
of the multiplication tensor of $\cA$ into $n$ elementary (symmetric) tensors.

As already discussed in the Introduction, multiplication algorithms are closely related to multiplicative linear secret sharing schemes.
In this context they are sometimes called ``multiplication-friendly embeddings'' \cite{CCCX09}.

A large part of the literature on multiplication algorithms (surveyed in \cite{BCPRRR19})
is devoted to the case where $\cA=\F_{q^m}[t]/(t^l)$ is a monogeneous local algebra over $\FF=\F_q$ a finite field.
This includes the cases $\cA=\F_{q^k}$ (finite field extension) and $\cA=\F_q[t]/(t^k)$ (truncated polynomials).
However, the case of a more general local algebra $\cA$ over an arbitrary field $\FF$,
for instance $\cA=\FF[[t_1,\dots,t_m]]/(t_1^l,\dots,t_m^l)$, is also of interest.

As for non-local algebras, the issue is related to Strassen's direct sum conjecture \cite{Strassen73}\cite{BrDo78}.
Indeed, such an algebra decomposes as the product of its localizations: $\cA=\prod_i\cA_i$.
The conjecture asserts that a multiplication algorithm for $\cA$ of minimal length should
decompose as a product (in the sense of~\ref{produit} below) of multiplication algorithms for each $\cA_i$. 

\item
More general multiplication reductions often occur as an intermediate step in the construction
of multiplication algorithms.
This includes generalized Chud\-nov\-sky-Chudnovsky type constructions that use evaluation-interpolation on curves
at points of higher degree and/or with multiplicities.
For instance (the symmetric part of) \cite[Th.~3.5]{HR-ChCh+} can be analyzed in two steps:
first, a multiplication \emph{reduction}
\beqv
\cA=\F_{q^m}[t]/(t^l)\;\leadsto\;\cB=\prod_i\F_{q^{d_i}}[t]/(t^{u_i})
\eeqv
is constructed, and then it is composed (in the sense of \ref{compose} below)
with the product of multiplication algorithms $\F_{q^{d_i}}[t]/(t^{u_i})\leadsto\F_q^{\mu_q^{\textrm{sym}}(d_i,u_i)}$,
to get a multiplication \emph{algorithm}
\beqv
\cA=\F_{q^m}[t]/(t^l)\;\leadsto\;\F_q^N
\eeqv
with $N=\sum_i\mu_q^{\textrm{sym}}(d_i,u_i)$.
\item
Last, reverse multiplication-friendly embeddings (RMFE), recently introduced in~\cite{CCXY18},
also are multiplication reductions, namely of the form
\beqv
\cA=\F_q^k\;\leadsto\;\cB=\F_{q^n}.
\eeqv
As such, the results presented in this work apply to them as well.
\end{itemize}

\begin{definition}
Let $\cB$ be an algebra over $\FF$. By a $\cB$-code we mean a linear subspace
\beqv
C\subset\cB.
\eeqv
If $C$ has dimension $k$ over $\FF$, we also say $C$ is a $[\cB,k]$-code.

In the particular case $\cB=\FF^n$, we say $C$ is a $[n,k]$-code, compatibly with the established literature.
\end{definition}

\entry
\label{adjoint}
Let $\phi:\cA\leadsto\cB$ be a multiplication reduction,
and consider its image $C_\phi=\phi(\cA)\subset\cB$.
It is a $[\cB,k]$-code, where $k=\dim\cA$.

If $x\in\cB$ is of the form $x=cc'$, where $c=\phi(a),c'=\phi(a')$ for $a,a'\in\cA$,
then necessarily one should have $\omega(x)=aa'$ in $\cA$.
By linearity, this condition \emph{uniquely determines}
the values of $\omega$ on the linear subspace
\beqv
C_\phi\deux=\linspan{\{cc'\,:\; c,c'\in C_\phi\}}\;\subset\cB
\eeqv
spanned by pairwise products of codewords (under $m_{\cB}$), called the \emph{square} of $C_\phi$.

By a slight abuse, the resulting uniquely determined map
\beqv
\omega:C_\phi\deux\longto\cA
\eeqv
will also be called the adjoint of $\phi$.

Conversely, values of $\omega$ outside of $C_\phi\deux$ do not matter in Definition~\ref{def-algo}.

This means we can take the commutativity of the diagram
\beqv
\begin{CD}
\cA\times\cA  @>m_{\cA}>>\cA\\
@V{\phi\times\phi}VV @AA{\omega}A\\
C_\phi\times C_\phi @>m_{\cB}>> C_\phi\deux
\end{CD}
\eeqv
as an alternative, equivalent definition for a multiplication reduction.
Indeed, once this holds, we can pick an arbitrary linear extension of $\omega$ from $C_\phi\deux$ to the whole
of $\cB$, and the condition in Definition~\ref{def-algo} will be satisfied.

It is customary in bilinear complexity theory to focus solely on the number
of bilinear operations in the algorithms, ignoring linear operations
as costless.
From this point of view, specifying the choice of $\omega$ in Definition~\ref{def-algo} would be somehow redundant.
However, if one is interested in practical implementations, it could happen
that some clever choice of an extension of $\omega$ from $C_\phi\deux$
to the whole of $\cB$ would give a better total complexity.


\entry
\label{equivalence}
Let $\phi:\cA\leadsto\cB$ be a multiplication reduction.
Given invertible elements $u\in\cA^\times$ and $v\in\cB^\times$,
define $\tilde{\phi}:\cA\to\cB$ by
\beqv
\tilde\phi(a)=v\phi(u^{-1}a)
\eeqv
for all $a\in\cA$.
Then $\tilde\phi$ is a multiplication reduction $\cA\leadsto\cB$.
Indeed, if $\phi$ has adjoint $\omega:C_\phi\deux\to\cA$,
then $\tilde\phi$ has adjoint $\tilde\omega:C_{\tilde\phi}\deux\to\cA$
where $\tilde\omega(b)=u^2\omega(v^{-2}b)$
for all $b\in C_{\tilde\phi}\deux=v^2C_\phi\deux\subset\cB$.

\begin{definition*}
We say two multiplication reductions $\phi,\tilde\phi:\cA\leadsto\cB$ are \emph{diagonally equivalent},\footnote{
one could devise another natural notion of equivalence, coarser, by further allowing composition
on the right with automorphisms of $\cA$ and on the left with automorphisms of $\cB$.}
and we write
\beqv
\phi\sim\tilde\phi,
\eeqv
if they can be related as above for some $u\in\cA^\times$, $v\in\cB^\times$.

This defines an equivalence relation on the set of multiplication reductions from $\cA$ to $\cB$.
\end{definition*}

\entry
\label{produit}
There is a notion of product for multiplication reductions: given $\phi_1:\cA_1\leadsto\cB_1$
and $\phi_2:\cA_2\leadsto\cB_2$,
we derive $\phi_1\times\phi_2:\cA_1\times\cA_2\leadsto\cB_1\times\cB_2$,
where each product algebra is equipped with componentwise multiplication.
If $\phi_1$ has adjoint $\omega_1$ and $\phi_2$ has adjoint $\omega_2$,
then $\phi_1\times\phi_2$ has adjoint $\omega_1\times\omega_2$.

This is compatible with diagonal equivalence:
if $\phi_1\sim\tilde\phi_1$
and $\phi_2\sim\tilde\phi_2$,
then $\phi_1\times\phi_2\,\sim\,\tilde\phi_1\times\tilde\phi_2$.

\entry
\label{compose}
Likewise, multiplication reductions can be composed (or concatenated):
given \mbox{$\phi:\cA\leadsto\cB$} and $\psi:\cB\leadsto\cC$,
we derive $\psi\circ\phi:\cA\leadsto\cC$.
If $\phi$ has adjoint $\omega$ and $\psi$ has adjoint $\varpi$,
then $\psi\circ\phi$ has adjoint $\omega\circ\varpi$.

However, in general, composition is \emph{not} compatible with diagonal equivalence,
i.e. it could happen that $\phi\sim\tilde\phi$ and $\psi\sim\tilde\psi$
but $\psi\circ\phi\,\not\sim\,\tilde\psi\circ\tilde\phi$.

\begin{proposition}\label{evaluation-interpolation}
Let $X$ be a scheme over $\FF$, 
together with points $Q\in X(\cA)$ and $\cP\in X(\cB)$ with coordinates in $\cA$ and $\cB$ respectively,
i.e. morphisms $Q:\Spec\cA\to X$ and $\cP:\Spec\cB\to X$.

Let $\mathcal{L}$ be an invertible sheaf on $X$, and $V\subset\Gamma(X,\mathcal{L})$ a finite dimensional linear system.
Let also $V\deux\subset\Gamma(X,\mathcal{L}^{\tens 2})$ be its square,
i.e. the linear system spanned by all pairwise products of elements of $V$.
Consider the natural restriction (i.e. pullback) maps from $V$ to $\cL|_Q$ and from $V\deux$ to $\cL^{\tens 2}|_{\cP}$, and assume
\begin{enumerate}[(i)]
\item $V\longto\cL|_Q$ is surjective
\item $V\deux\longto\cL^{\tens 2}|_{\cP}$ is injective.
\end{enumerate}
Then there is a multiplication reduction
\beqv
\phi:\cA\leadsto\cB
\eeqv
given by the following construction:
\begin{itemize}
\item
choose trivializations $\cL|_Q\simeq\cA$ and $\cL|_{\cP}\simeq\cB$ (recall $\cA,\cB$ finite over $\FF$)
\item
choose a right-inverse $\sigma:\cA\to V$ of the surjective map $V\to\cL|_Q\simeq\cA$
\item
compose $\sigma$ with the ``evaluation-at-$\cP$'' map $V\to\cL|_{\cP}\simeq\cB$,
to get the linear map $\phi:\cA\to\cB$.
\end{itemize}
Moreover, if we assume that $V\to\cL|_Q$ is \emph{bijective},
then the diagonal equivalence class of $\phi$ so constructed is \emph{independent} of the choices made.
\end{proposition}
(Observe that if we're interested in multiplication algorithms, then $\cB=\FF^n$,
and we can see $\cP$ as an ordered collection $\cP=(P_1,\dots,P_n)$ of $n$ rational points $P_i\in X(\FF)$. This explains our notation.
But on the other hand, if we're interested in RMFE $\cA=\F_q^k\leadsto\cB=\F_{q^n}$, then $Q=(Q_1,\dots,Q_k)$ is a collection of $k$ rational points,
while $\cP$ is a point defined over an extension.
Our formalism encompasses both situations, and even more general ones,
in a uniform way.) 
\begin{proof}
Given $f\in V$, write $f(Q)$ and $f(\cP)$ for the images of $f$ under the evaluation maps $V\to\cL|_Q\simeq\cA$ and $V\to\cL|_{\cP}\simeq\cB$, respectively.\footnote{e.g. when $\cB=\FF^n$ and $\cP=(P_1,\dots,P_n)$, we can write $f(\cP)=(f(P_1),\dots,f(P_n))\in\FF^n$.}

Likewise, given $h\in V\deux$, write $h(Q)$ and $h(\cP)$ for the images of $h$ under $V\deux\to\cL^{\tens 2}|_Q\simeq\cA$ and $V\to\cL^{\tens 2}|_{\cP}\simeq\cB$,
where the trivializations $\cL^{\tens 2}|_Q\simeq\cA$ and $\cL^{\tens 2}|_{\cP}\simeq\cB$ are deduced from
the ones chosen previously by passing to the square.
This makes evaluation commute with multiplication:
for $f,f'\in V$ we have $(ff')(Q)=f(Q)f'(Q)$ in $\cA$,
and likewise $(ff')(\cP)=f(\cP)f'(\cP)$ in $\cB$.

Now we construct the adjoint linear map $\omega:\cB\to\cA$,
and then we check that $\phi,\omega$ indeed satisfy the condition in Definition~\ref{def-algo}.

Since $V\deux\to\cL^{\tens 2}|_{\cP}\simeq\cB$ is injective, we can choose a left-inverse $\rho:\cB\to V\deux$.
We then compose $\rho$ with the evaluation-at-$Q$ map $V\deux\to\cL^{\tens 2}|_Q\simeq\cA$
to get a linear map $\omega:\cB\to\cA$.

Let $a,a'\in\cA$. By construction, $\sigma(a),\sigma(a')$ are elements $f_a,f_{a'}\in V$
with $f_a(Q)=a$, $f_{a'}(Q)=a'$, and then $\phi(a)=f_a(\cP)$, $\phi(a')=f_{a'}(\cP)$ in~$\cB$.
Since evaluation and multiplication commute,
we get $\phi(a)\phi(a')=f_a(\cP)f_{a'}(\cP)=(f_af_{a'})(\cP)$.
By construction, $\rho((f_af_{a'})(\cP))$ is then an element $h\in V\deux$ such that $h(\cP)=(f_af_{a'})(\cP)$.
By injectivity of $V\deux\to\cL^{\tens 2}|_{\cP}$, this forces $h=f_af_{a'}$.
And then, commuting evaluation and multiplication again,
we get $\omega(\phi(a)\phi(a'))=(f_af_{a'})(Q)=f_a(Q)f_{a'}(Q)=aa'$, as desired.

As for the last statement, observe that in the first step of the construction,
another choice of trivializations multiplies the evaluation-at-$Q$ and -at-$\cP$ maps
by some $u\in\cA^\times$ and $v\in\cB^\times$, respectively.
Moreover if $V\to\cL|_Q$ is bijective, then in the second step there is only one choice
for the right inverse, which is actually the inverse of the evaluation-at-$Q$ map.
All in all this replaces $\phi$ with an equivalent map precisely as in~\ref{equivalence}.
\end{proof}

\begin{definition}
A multiplication reduction $\phi:\cA\leadsto\cB$ is geometric,
or is of evaluation-interpolation type,
if it can be obtained through Construction~\ref{evaluation-interpolation}.

Conversely, in this situation, we say that $(X,Q,\cP,\cL,V)$ 
is a \emph{geometric realization} of the reduction $\phi$.
\end{definition}

This can be compared with the generalized Chudnovsky-Chudnovsky constructions
surveyed in~\cite{BCPRRR19}, that use evaluation-interpolation on curves at points of higher degree
and/or with multiplicities.
This fits in our formalism of multiplication reductions. But if we insist on having a multiplication \emph{algorithm},
so $\cB=\FF^n$, the key point then is that we ask for evaluation-interpolation at $\cP=(P_1,\dots,P_n)$,
a collection of \emph{rational points, without multiplicities}.
But on the other hand, we allow $X$ to be a higher dimensional variety,
or even an arbitrary scheme.

\begin{remark}
Some arguments suggest that ``good'' geometric multiplication reductions should actually be low dimensional. 

For instance, suppose given an algebra $\cA$, of dimension $k$, and consider a geometric multiplication
algorithm $\phi:\cA\leadsto\FF^n$, of length $n$ as small as possible.
Then condition (i) forces $\dim V\geq k$, and condition (ii) forces $\dim V\deux\leq n$.
Thus, if $X$ is a variety (i.e. an \emph{integral} scheme), and if we assume
\beqv
n\leq 3k-4,
\eeqv
then the field analogue \cite{BCZ18} of Freiman's theorem applies
and shows that the subfield generated by $V^{-1}V\subset\FF(X)$ is the function field $\FF(Y)$
of a curve $Y$.
With a little extra work, one could then prove that $\phi$ can be realized over
this curve $Y$.

Interestingly, if $\FF$ is the finite field $\F_p$ for $p\geq7$ prime, and $\cA=\F_{p^k}$ is its degree $k$ field
extension for $k$ large, then the best constructions \cite[\S6]{HR-ChCh+}
allow to get $n\approx 3\left(1+\frac{2}{p-2}\right)k$, which is slightly above this $3k-4$ bound.
So the argument does not apply there.
\end{remark}

\entry
The notion of \emph{supercode} introduced in \cite{STV92} for multiplication algorithms
is easily extended in the context of multiplication reductions.
We define a supercode as a $\FF$-linear subspace
\beqv
\widehat{C}\subset\cA\times\cB
\eeqv
such that
the first projection $\pi_{\cA}:\widehat{C}\to\cA$ is surjective,
and the second projection $\pi_{\cB}:\widehat{C}\deux\to\cB$ is injective,
where as usual $\widehat{C}\deux$ is the $\FF$-linear span of
pairwise products of elements of $\widehat{C}$ in the product algebra $\cA\times\cB$.

If $\phi:\cA\leadsto\cB$ is a multiplication reduction,
then $\widehat{C}=\im(\id_{\cA}\times\phi)\subset\cA\times\cB$ is a supercode.

\entry
\label{trivialrepr}
Conversely, if $\widehat{C}\subset\cA\times\cB$ is a supercode,
then $\widehat{C}$ gives rise to a multiplication reduction. This can be seen in our geometric formalism.

Set $X=\Spec(\cA\times\cB)$,
and consider the points $Q\in X(\cA)$ given by the first projection $\cA\times\cB\to\cA$,
and $\cP\in X(\cB)$ given by the second projection $\cA\times\cB\to\cB$.

Also set $\cL=\cO_X$, the structure sheaf,
and consider then $\widehat{C}$ as an incomplete linear system $\widehat{C}\subset\Gamma(X,\cO_X)=\cA\times\cB$.
\begin{proposition*}
With these notations, $(X,Q,\cP,\cO_X,\widehat{C})$ is the geometric realization of a multiplication reduction $\cA\leadsto\cB$.

Moreover, if $\widehat{C}=\im(\id_{\cA}\times\phi)$ comes from a multiplication reduction $\phi:\cA\leadsto\cB$,
then actually $(X,Q,\cP,\cO_X,\widehat{C})$ is a geometric realization for $\phi$.
\end{proposition*}
\begin{proof}
Proposition~\ref{evaluation-interpolation} applies, since the surjectivity
and injectivity conditions in the definition of a supercode reflect conditions (i) and (ii).

Moreover, if $\widehat{C}=\im(\id_{\cA}\times\phi)$, then we have bijectivity in (i),
and lifting the first projection from $\cA$ back to $\widehat{C}$ and projecting to $\cB$ indeed gives $\phi$, by construction.
\end{proof}

It follows at once:

\begin{corollary}\label{cortauto}
Any multiplication reduction is geometric.
\end{corollary}

Unfortunately, this result, based on the tautological, zero-dimensional space $\Spec(\cA\times\cB)$, is quite disappointing.\footnote{
A similarly tautological statement is that any linear code $C$ is the evaluation
code, at the $n$ points of $\Spec(\FF^n)$, of $C$ itself, seen as an incomplete
linear system in $\Gamma(\Spec(\FF^n),\cO_{\Spec(\FF^n)})$.}
One aim of this work will be to present geometric realizations with a more interesting structure.

This will allow us to also address the following question:
how can one tell if a given linear map $\phi:\cA\to\cB$ is a multiplication reduction?
Or instead of specifying the whole map $\phi$, one can give only its image $C$.
The question then reduces to the following, from~\cite{STV92}: how can one tell if a given code $C\subset\cB$
is the second projection of a supercode $\widehat{C}\subset\cA\times\cB$?
To these questions we will provide a geometric answer, albeit perhaps not an algorithmic criterion, as could be even more desirable.

\section{An example in dimension $2$}
\label{dP}

\entry
Here we work over the base field $\FF=\F_2$.
We will construct a geometric multiplication algorithm
\beqv
\F_{2^5}\leadsto\F_2^{13}
\eeqv
by evaluation-interpolation on a surface.

It turns out this length $n=13$ is best possible for a multiplication algorithm
for $\F_{2^5}$ over $\F_2$, i.e. it reaches the \emph{(symmetric) bilinear complexity} 
\beqv
\mu_2^{\textrm{sym}}(5)=13.
\eeqv

Indeed, no such algorithm exists in length $12$ or less, because otherwise,
it would imply the existence of a $[12,5,5]$ binary linear code \cite[Cor~5.1]{ChCh88}\cite[Prop.~1.4]{STV92},
which is known to be false \cite{FoPe59}\cite{codetables}.

\entry
A smooth conic $\cQ$ in $\PP^2$ over $\F_2$ contains precisely three rational points $P_1,P_2,P_3$
and one pair of conjugate quadratic points $Q,\overline{Q}$,
no three of which lie on a line.

Blowing-up $\PP^2$ at these five points gives a (nonsingular) del Pezzo surface $X$ of degree $4$ \cite[Th.~24.4]{Manin},
with
\beqv
\abs{X(\F_2)}=13.
\eeqv
More precisely, there are $3$ rational points on each of the exceptional divisors $E_1,E_2,E_3$ above $P_1,P_2,P_3$,
and then $4$ other rational points above the remaining points $P_4,P_5,P_6,P_7$ in $\PP^2(\F_2)$.

The canonical sheaf of $X$ is
\beqv
\cK_X=(\pi^*\cO(-3))(E_1+E_2+E_3+F+\overline{F})
\eeqv
where $\pi:X\to\PP^2$ is the blow-up map, and $F,\overline{F}$ are the exceptional divisors over $Q,\overline{Q}$, respectively.

The anticanonical embedding is projectively normal \cite[Th.~8.3.4]{Dolgachev} and realizes $X$ as an intersection of two quadrics in $\PP^4$ \cite[Th.~8.6.2]{Dolgachev},
thus:
\begin{itemize}
\item $\dim\Gamma(X,\cK_X^{-1})=5$
\item $\dim\Gamma(X,\cK_X^{-2})=13$.
\end{itemize}
Here $\pi_*$ gives an identification
\beqv
\Gamma(X,\cK_X^{-1})=\Gamma(\PP^2,\fI_{P_1,P_2,P_3,Q,\overline{Q}}\cO(3))
\eeqv
of the space of sections of the anticanonical sheaf with the linear system of cubic forms on $\PP^2$ vanishing at $P_1,P_2,P_3,Q,\overline{Q}$
(which, indeed, has dimension $5$, e.g. by Cayley-Bacharach).

Likewise we have an identification
\beqv
\Gamma(X,\cK_X^{-2})=\Gamma(\PP^2,\fI_{P_1,P_2,P_3,Q,\overline{Q}}^2\cO(6))
\eeqv
where the right-hand side is the linear system of sextic forms on $\PP^2$ vanishing at $P_1,P_2,P_3,Q,\overline{Q}$ with multiplicity at least $2$.

\begin{lemma}
\label{dPinj}
The natural restriction map, from $\Gamma(X,\cK_X^{-2})$ to the product of the fibers of $\cK_X^{-2}$ at the $13$ rational points of $X$, is injective.
\end{lemma}
\begin{proof}
Let $s\in\Gamma(X,\cK_X^{-2})$ vanishing at the $13$ rational points of $X$. We have to show $s=0$.

First, we use the vanishing of $s$ at the $9$ rational points in $E_1,E_2,E_3$.
Since $\cK_X^{-2}=(\pi^*\cO(6))(-2E_1-2E_2-2E_3-2F-2\overline{F})$,
the restriction $s|_{E_i}$ defines a section of
\beqv
\cO(-2E_i)|_{E_i}\simeq\cO_{\PP^1}(2).
\eeqv
The fact that this degree $2$ section vanishes at the $3$ rational points of $E_i\simeq\PP^1$
then forces $s|_{E_i}=0$, so actually $s$ is a section of
\beqv
\cK_X^{-2}(-E_1-E_2-E_3)=(\pi^*\cO(6))(-3E_1-3E_2-3E_3-2F-2\overline{F}).
\eeqv
This means $f=\pi_*s\in\Gamma(\PP^2,\cO(6))$ vanishes with multiplicity at least $3$ at $P_1,P_2,P_3$ and $2$ at $Q,\overline{Q}$.

In turn, the vanishing of $s$ at the remaining $4$ rational points of $X$
means that $f$ vanishes at their images $P_4,P_5,P_6,P_7$ in $\PP^2$.

Recollecting, we have established so far that $f=\pi_*s\in\Gamma(\PP^2,\cO(6))$ is a sextic form that vanishes:
\begin{itemize}
\item at $P_1,P_2,P_3$ with multiplicity at least $3$
\item at $Q,\overline{Q}$ with multiplicity at least $2$
\item at $P_4,P_5,P_6,P_7$ with multiplicity at least $1$.
\end{itemize}
We have to show $f=0$.
\vspace{.5\baselineskip}

On the conic $\cQ=(P_1P_2P_3Q\overline{Q})$, the sextic $f$ vanishes at $3P_1+3P_2+3P_3+2Q+2\overline{Q}$,
a divisor of total degree $13$. By B\'ezout, this forces $f|_{\cQ}=0$, i.e. $f=qt$
where $q\in\Gamma(\PP^2,\cO(2))$ is an equation for $\cQ$,
and $t\in\Gamma(\PP^2,\cO(4))$ is a quartic that has to vanish further:
\begin{itemize}
\item at $P_1,P_2,P_3$ with multiplicity at least $2$
\item at $Q,\overline{Q}$ and $P_4,P_5,P_6,P_7$ with multiplicity at least $1$.
\end{itemize}

Set $L=(Q\overline{Q})$. It is a rational line that does not meet $P_1,P_2,P_3$
(and actually it is the only such line).
Relabelling, we may assume $P_4=L\cap(P_1P_2)$, $P_5=L\cap(P_1P_3)$, $P_6=L\cap(P_2P_3)$. 
\vspace{.5\baselineskip}

Then, the quartic $t$ vanishes at a divisor of degree~$5$ on each of the following lines:
\begin{itemize}
\item at $2P_1+2P_2+P_4$ on $(P_1P_2)$
\item at $2P_1+2P_3+P_5$ on $(P_1P_3)$
\item at $2P_2+2P_3+P_6$ on $(P_2P_3)$
\item at $Q+\overline{Q}+P_4+P_5+P_6$ on $L$.
\end{itemize}
This implies that, up to multiplication by a constant, $t$ is the product of the equations of these four lines.
However, $t$ also vanishes at $P_7$, while none of these lines passes through $P_7$.
Thus the only possibility is $t=0$, hence $f=0$, $s=0$.
\end{proof}

\entry
\label{dPsurj}
Now consider another smooth conic $\cQ'$, that passes through $P_1,P_2$ but not through $P_3,Q,\overline{Q}$.

For instance, setting $L=(Q\overline{Q})$, $P_4=L\cap(P_1P_2)$, $P_5=L\cap(P_1P_3)$,
and letting $Q',\overline{Q'}$ be the conjugate quadratic points of the line $(P_3P_4)$,
we can take $\cQ'=(P_1P_2P_5Q'\overline{Q'})$.

Let then $R_1\in\cQ'(\F_{2^5})$ be a point of degree $5$ on $\cQ'$, non-rational (there are $30$ of them),
and write $R_2,R_3,R_4,R_5$ for its conjugates.

\begin{lemma*}
The natural restriction map $\Gamma(X,\cK_X^{-1})\to\cK_X^{-1}|_{R_1}$ is surjective.
\end{lemma*}
\begin{proof}
Since both sides have dimension $5$, it suffices to prove that this restriction map is injective.
So consider a section $s\in\Gamma(X,\cK_X^{-1})$ that vanishes at $R_1$. We have to show $s=0$.

Assume the contrary. Then $f=\pi_*s\in\Gamma(\PP^2,\cO(3))$ defines a cubic $\cC\subset\PP^2$
that passes through $P_1,P_2,P_3,Q,\overline{Q}$,
and also through $R_1$ and its conjugates $R_2,R_3,R_4,R_5$, because of Galois invariance.
Then $\cC\cap\cQ'$ contains the $7$ points $P_1,P_2,R_1,R_2,R_3,R_4,R_5$, so by B\'ezout
we have $\cC=\cQ'\cup D$ for a certain line $D$.
Then $D$ has to pass through $P_3,Q,\overline{Q}$, a contradiction since these three points do not lie on a line.
\end{proof}

\begin{proposition}
\label{algo_13_5}
With these notations, $(X,R_1,X(\F_2),\cK_X^{-1},\Gamma(X,\cK_X^{-1}))$ is the geometric realization
of a multiplication algorithm $\F_{2^5}\leadsto\F_2^{13}$.
\end{proposition}
\begin{proof}
The surjectivity condition (i) in Proposition~\ref{evaluation-interpolation} is satisfied by~\ref{dPsurj},
and the injectivity condition (ii), by~\ref{dPinj};
actually, both evaluation maps in (i) and (ii) turn out to be bijective here.
\end{proof}

\entry
All this can be made explicit.
Set $\F_4=\F_2[\alpha]/(\alpha^2+\alpha+1)$ and $\F_{2^5}=\F_2[\gamma]/(\gamma^5+\gamma^2+1)$. Let $x,y,z$ be linear coordinates on $\PP^2$.

Take $P_1=(1:0:0)$, $P_2=(0:1:0)$, $P_3=(0:0:1)$, $L=\{x+y+z=0\}$, $Q=(1:\alpha:\alpha^2)$, $\overline{Q}=(1:\alpha^2:\alpha)$, $\cQ=\{xy+xz+yz=0\}$.

Then $P_4=(1:1:0)$, $P_5=(1:0:1)$, $(P_3P_4)=\{x=y\}$, $Q'=(1:1:\alpha)$, $\overline{Q'}=(1:1:\alpha^2)$, $\cQ'=\{xy+xz+z^2=0\}$, and we can take $R_1=(\gamma^2:\gamma^3:1)$.

As a basis for $V=\Gamma(X,\cK_X^{-1})=\Gamma(\PP^2,\fI_{P_1,P_2,P_3,Q,\overline{Q}}\cO(3))$
we take:
\begin{align*}
v_0&=x(xy+xz+yz)\\
v_1&=y(xy+xz+yz)\\
v_2&=xy(x+y+z)+x(xy+z^2)\\
v_3&=(x+z)(xy+xz+yz)\\
v_4&=x(xy+z^2).
\end{align*}
Indeed, these cubics are easily seen to vanish at $P_1,P_2,P_3,Q,\overline{Q}$,
and were chosen through linear algebra so as to evaluate at $R_1$ as
\beqv
v_i(R_1)=\gamma^i
\eeqv
for $i=0,\dots,4$.

Now we evaluate these $v_i$ at the $13$ points in $X(\F_2)$ to get $\phi$.
Evaluation at $P_4,P_5,P_6,P_7$ is straightforward,
but how to evaluate $v_i$ at the points that lie in an exceptional divisor, say, at the three points in $E_1$ above $P_1$?
These values are given by the derivatives $(\partial_yv_i)(P_1)$, $(\partial_zv_i)(P_1)$, and $((\partial_y+\partial_z)v_i)(P_1)$.

Doing so, we find that $\phi:\F_{2^5}\leadsto\F_2^{13}$
is given by the matrix
\beqv
\left(
\begin{array}{ccccccccccccc}
1 & 1 & 0 & 0 & 0 & 0 & 0 & 0 & 0 & 1 & 1 & 0 & 1\\
0 & 0 & 0 & 1 & 1 & 0 & 0 & 0 & 0 & 1 & 0 & 1 & 1\\
1 & 1 & 0 & 1 & 0 & 1 & 0 & 1 & 1 & 0 & 1 & 1 & 1\\
1 & 1 & 0 & 0 & 0 & 0 & 1 & 1 & 0 & 1 & 0 & 1 & 0\\
1 & 0 & 1 & 0 & 0 & 0 & 1 & 0 & 1 & 1 & 1 & 0 & 0
\end{array}
\right)
\eeqv
relative to the basis $1,\gamma,\gamma^2,\gamma^3,\gamma^4$ of $\F_{2^5}$ (with row vector notation).

As observed in~\ref{adjoint}, the adjoint $\omega:\F_2^{13}\to\F_{2^5}$ of $\phi$ is then uniquely determined,
by mere linear algebra (i.e. we need not follow the proof of Proposition~\ref{evaluation-interpolation}).
We find it is given by (the transpose of) the matrix
\beqv
\left(
\begin{array}{ccccccccccccc}
1 & 0 & 0 & 0 & 1 & 1 & 0 & 1 & 1 & 1 & 0 & 1 & 1\\
1 & 1 & 0 & 1 & 0 & 0 & 1 & 0 & 1 & 0 & 1 & 0 & 1\\
0 & 1 & 0 & 0 & 0 & 0 & 0 & 1 & 1 & 1 & 1 & 1 & 1\\
1 & 0 & 1 & 1 & 1 & 1 & 1 & 1 & 1 & 0 & 1 & 0 & 0\\
1 & 1 & 0 & 1 & 0 & 1 & 0 & 1 & 1 & 0 & 0 & 1 & 0
\end{array}
\right)
\eeqv
so $\omega_1=1+\gamma+\gamma^3+\gamma^4$, $\,\omega_2=\gamma+\gamma^2+\gamma^4$, $\,\omega_3=\gamma^3$, etc. 

\begin{remark}
The space of quadratic relations between $v_0,v_1,v_2,v_3,v_4$
is generated by
\beqv
v_0^2+v_0v_2+v_1v_4
\eeqv
and
\beqv
v_0(v_0+v_1+v_3+v_4)+v_3(v_2+v_3+v_4).
\eeqv
Thus, now seeing $v_0,v_1,v_2,v_3,v_4$ as indeterminates that define linear coordinates in $\PP^4$,
we find that $X\subset\PP^4$ is the surface cut out by these two quadrics.
\end{remark}

\begin{remark}
The binary linear code $C_\phi$ image of $\phi$
has parameters $[13,5,5]$, which is optimal.
Its automorphism group has size $48$;
it acts transitively on the first $12$ coordinates,
but fixes the last coordinate (which corresponds to evaluation at $P_7=(1:1:1)$).

The square code $C_\phi\deux$ is the whole of $\F_2^{13}$.
\end{remark}

\section{The canonical point}

\entry
Following the Russian school of coding theory, any (non-degenerate) linear code can be
obtained by evaluation of linear functions on a system of rational points in a projective
space. Here we consider what happens when we apply this point of view to multiplication algorithms,
or more generally to multiplication reductions.

If $V$ is a finite dimensional $\FF$-vector space, recall that $\PP V=\Proj S^\cdot V$
is the (``dual'') projective space that parameterizes invertible quotients of $V$.
More precisely, for any algebra $\cR$, there is a natural bijection between
the set $\PP V(\cR)$ of points of $\PP V$ with coordinates in $\cR$,
and the set of equivalence classes of surjective maps
\beqv
q:V\tens_{\FF}\cR\surj L
\eeqv
with $L$ a locally free $\cR$-module of rank~$1$.

In coordinates, the choice of a basis $v_1,\dots,v_k$ of $V$ identifies
the symmetric algebra $S^\cdot V$ with the polynomial algebra $\FF[x_1,\dots,x_k]$
(where $v_i\leftrightarrow x_i$),
hence it identifies $\PP V$ with the standard projective space $\PP^{k-1}$.
If moreover $\cR$ is finite over $\FF$, then $L$ is actually free.
Choosing a trivialization $L\simeq\cR$, we see that the quotient map $q:V\tens_{\FF}\cR\surj L\simeq\cR$
then corresponds to the point
\beqv
(q(v_1\tens 1):\dots:q(v_k\tens 1))
\eeqv
in $\PP^{k-1}(\cR)$.

\entry
Let now $\cA$ be an algebra, finite over $\FF$.
We apply what precedes with $V=\cR=\cA$ and the multiplication map
\beqv
m_{\cA}:\cA\tens_{\FF}\cA\surj\cA.
\eeqv
\begin{definition*}
The \emph{canonical point} of $\cA$ is the point
\beqv
Q_{\cA}\in\PP\cA(\cA)
\eeqv
corresponding to $m_{\cA}$.
\end{definition*}
In coordinates, if we choose a basis $a_1,\dots,a_k$ of $\cA$,
identifying $\PP\cA$ with $\PP^{k-1}$ accordingly, then
\beqv
Q_{\cA}=(a_1:\dots:a_k)\in\PP^{k-1}(\cA).
\eeqv

\entry
\label{PC}
We say that a $\cB$-code $C\subset\cB$ is \emph{non-degenerate} if it generates the unit ideal in $\cB$:
\beqv
\cB\cdot C=\cB.
\eeqv
(For instance, if $C\subset\FF^n$ is a linear code in the usual sense, it means $C$ has full support.)

If $C$ is non-degenerate, then the surjective map
\beqv
m_{\cB}:\;C\tens\cB\,\surj\,C\cdot\cB=\cB
\eeqv
defines a point
\beqv
\cP_C\in\PP C(\cB),
\eeqv
which we call the \emph{projective system} of $C$.

In coordinates, let $c_1,\dots,c_k$ be a basis of $C$, and identify $\PP C$ with $\PP^{k-1}$ accordingly.
Then
\beqv
\cP_C=(c_1:\dots:c_k)\in\PP^{k-1}(\cB).
\eeqv
Furthermore, write $\cB=\cB_1\times\cdots\cB_m$ as a product of local algebras, and for $1\leq j\leq m$
let $\pi_j:\cB\to\cB_j$ be the $j$-th projection.
This allows to define the generator matrix of $C$ (with respect to the given basis)
as the $k\times m$ matrix $\GG$ whose $(i,j)$-th entry is $\pi_j(c_i)\in\cB_j$.
Then we can view $\cP_C$ as an ordered collection
\beqv
\cP_C=(P_1,\dots,P_m)
\eeqv
where
\beqv
P_j=(\pi_j(c_1):\dots:\pi_j(c_k))\in\PP^{k-1}(\cB_j)
\eeqv
is the point defined by the $j$-th column of $\GG$.
(For instance, if $\cB=\FF^n$, then $\cP_C$ is the projective system of rational points associated to $C$
in the usual sense.)

Then we retrieve $C$ as the image of the evaluation-at-$\cP_C$ map
\beqv
\Gamma(\PP^{k-1},\cO(1))\longto\cB
\eeqv
which sends a linear function $l\in\Gamma(\PP^{k-1},\cO(1))=\FF[x_1,\dots,x_k]_1$ to its value $l(c_1,\dots,c_k)\in\cB$,
or equivalently to $(l(P_1),\dots,l(P_m))\in\cB_1\times\cdots\cB_m$.

\entry
\label{Pphi}
We also say a linear map $\phi:\cA\to\cB$ is non-degenerate if it is injective
and its image $C_\phi=\phi(\cA)\subset\cB$ is a non-degenerate $\cB$-code.
For instance:
\begin{itemize}
\item If $\phi:\cA\leadsto\FF^n$ is a multiplication algorithm,
it means its components $\phi_1,\dots,\phi_n\in\cA^\vee$ are all nonzero.
We can always assume this holds: otherwise, just discard the bad coordinates to get
a shorter, non-degenerate algorithm.
\item
A reverse multiplication-friendly embedding $\phi:\F_q^k\leadsto\F_{q^n}$ is always non-degenerate,
since it is nonzero and its image lives in $\F_{q^n}$, a field.
\item
For a general multiplication reduction $\phi:\cA\leadsto\cB$, write $\cB=\cB_1\times\cdots\times\cB_m$ as
a product of local algebras, and let $\phi_j:\cA\to\cB_j$ be the components of $\phi$.
Then $\phi$ non-degenerate means for all $j$, the image $\phi_j(\cA)$ contains an invertible
element in $\cB_j^\times$.
\end{itemize}
Since $\phi$ is injective, it identifies $\cA$ with its image: $\cA\simeq C_{\phi}$,
so it identifies the corresponding projective spaces: $\PP\cA\simeq\PP C_{\phi}$.
We then let
\beqv
\cP_\phi\in\PP\cA(\cB)
\eeqv
correspond to $\cP_{C_\phi}\in\PP C_\phi(\cB)$ under this identification.

Then we retrieve $\phi$ as the evaluation-at-$\cP_\phi$ map
\beqv
\cA=\Gamma(\PP\cA,\cO(1))\;\longto\;\cB.
\eeqv
We also call $\cP_\phi$ the projective system of $\phi$.

In coordinates, let $a_1,\dots,a_k$ be a basis of $\cA$, and identify $\PP\cA$ with $\PP^{k-1}$ accordingly.
Then
\beqv
\cP_\phi=(\phi(a_1):\dots:\phi(a_k))\in\PP^{k-1}(\cB).
\eeqv
Furthermore, write $\cB=\cB_1\times\cdots\cB_m$ as a product of local algebras, and for $1\leq j\leq m$
let $\phi_j:\cA\to\cB_j$ be the $j$-th component of $\phi$.
Then we can view $\cP_\phi$ as an ordered collection
\beqv
\cP_\phi=(P_1,\dots,P_m)
\eeqv
where
\beqv
P_j=(\phi_j(a_1):\dots:\phi_j(a_k))\in\PP^{k-1}(\cB_j).
\eeqv

\entry
\label{almost}
Let $\phi:\cA\leadsto\cB$ be a non-degenerate multiplication reduction.
Let $Q_{\cA}\in\PP\cA(\cA)$ be the canonical point of $\cA$,
and let $\cP_\phi\in\PP\cA(\cB)$ as introduced just above.
Also consider the complete linear system $V=\Gamma(\PP\cA,\cO(1))=\cA$.

Then,
\beqv
(\PP\cA,Q_{\cA},\cP_\phi,\cO(1),V)
\eeqv
is ``almost'' a geometric realization
for $\phi$.

By this, we mean that $\phi$ indeed results from the construction given in Proposition~\ref{evaluation-interpolation} applied to this geometric data.
This is because evaluation of $V=\cA$ at $Q_{\cA}$ is the identity map of $\cA$,
and evaluation of $V=\cA$ at $\cP_\phi$ is $\phi$.

However, to have a geometric realization proper, one should check conditions (i) and (ii).
Condition (i) is clearly satisfied, but turning to condition (ii),
for the complete linear system $V=\Gamma(\PP\cA,\cO(1))$
we have $V\deux=\Gamma(\PP\cA,\cO(2))=S^2\cA$, which might fail to evaluate
injectively at $\cP_\phi$.

More precisely, the defect of injectivity is measured by the kernel of the
map
\beqv
S^2\cA\overset{S^2\phi}{\simeq}S^2C_\phi\overset{m_{\cB}\,}{\longto}C_\phi\deux\subset\cB,
\eeqv
which will be studied below.

\section{The quadratic hull}

\entry
\label{I2f}
As explained in the Introduction, the quadratic kernel $I_2(C)$ of a code
was first considered in \cite[sec.~3]{MMP11}, and also later in \cite[\S1.11]{HR-AGCT14}.

In full generality, if $f:Y\to\PP$ is a morphism from a scheme to a projective space over $\FF$,
we define $I_2(f)$ as the degree $2$ part of the homogeneous ideal of the (schematic) image of $f$.
Equivalently:
\beqv
I_2(f)=\ker(\Gamma(\PP,\cO(2))\overset{f^*}{\longto}\Gamma(Y,f^*\cO(2))).
\eeqv
And then:
\begin{definition}
The quadratic hull of $f$ is the closed subscheme
\beqv
Z_2(f)=Z(I_2(f))\subset\PP
\eeqv
defined by $I_2(C)$.

Equivalently, $Z_2(f)$ is the (schematic) intersection of the quadratic hypersurfaces through which $f$ factorizes.

It is the smallest closed subscheme of $\PP$ that contains the (schematic) image of $f$ and is defined by quadratic equations. 
\end{definition}

\entry 
If $C\subset\cB$ is a non-degenerate $\cB$-code,
we set $I_2(C)=I_2(\cP_C)$ and define the quadratic hull of $C$ as
\beqv
Z_2(C)=Z_2(\cP_C)\subset\PP C,
\eeqv
where $\cP_C$ from~\ref{PC} is seen as a morphism $\cP_C:\Spec\cB\to\PP C$.

In this situation, \ref{I2f} specializes to
\beqv
\begin{split}
I_2(C)&=\ker(S^2C\overset{m_{\cB}\,}{\longto}\cB)\\
&=\ker(S^2C\overset{m_{\cB}\,}{\longto}C\deux).
\end{split}
\eeqv
And thus, the quadratic hull of $C$ is the closed subscheme
\beqv
Z_2(C)=Z(I_2(C))\subset\PP C
\eeqv
defined by this $I_2(C)$.

In coordinates, as in~\ref{PC}, let $c_1,\dots,c_k$ be a basis of $C$,
write $\cB=\cB_1\times\cdots\cB_m$ as a product of local algebras, let $\GG$ be the corresponding $k\times m$
generator matrix for $C$, and for $1\leq j\leq m$ let $P_j\in\PP^{k-1}(\cB_j)$ be the point defined by the $j$-th column of $\GG$.
Then $C\deux$ is the image of the evaluation-at-$\cP_C$ map
\beqv
\Gamma(\PP^{k-1},\cO(2))\longto\cB
\eeqv
which sends a homogeneous quadratic form $q\in\Gamma(\PP^{k-1},\cO(2))=\FF[x_1,\dots,x_k]_2$ to its value $q(c_1,\dots,c_k)=(q(P_1),\dots,q(P_m))$,
and $I_2(C)$ is the kernel of this evaluation map:
\beqv
\begin{split}
I_2(C)&=\{q\in\FF[x_1,\dots,x_k]_2:\;q(c_1,\dots,c_k)=0\}\\
&=\{q\in\FF[x_1,\dots,x_k]_2:\;q(P_1)=0,\dots,q(P_m)=0\}.
\end{split}
\eeqv
And then $Z_2(C)=Z_2(\{P_1,\dots,P_m\})$ is the intersection of all quadrics in $\PP^{k-1}$ that contain $P_1,\dots,P_m$.
This intersection should be understood scheme-theoretically,
so in general $Z_2(C)$ need not be a nice smooth, irreducible, nor even reduced variety.

\begin{example}
Let $\FF$ be a field in which $2$ admits two square roots (e.g. $\FF=\F_7$, $\pm\sqrt{2}=\pm3$).
Let $C\subset\FF^7$ be the linear code with generator matrix
\beqv
\GG=\left(
\begin{array}{ccccccc}
\makebox[5mm]{1} & \makebox[5mm]{1} & \makebox[5mm]{1} & \makebox[5mm]{1} & \makebox[5mm]{1} & \makebox[5mm]{1} & \makebox[5mm]{1}\\
0 & 1 & 1 & 0 & 0 & 1 & 1\\
0 & 0 & 0 & 1 & 1 & 1 & 1\\
0 & 1 & -1 & 1 & -1 & \!\!\sqrt{2}\!\! & \!\!-\sqrt{2}\!\!
\end{array}
\right)
\eeqv
and let $P_1,\dots,P_7\in\PP^3(\FF)$ be the points defined by the columns of $\GG$.

If $w,x,y,z$ are linear coordinates on $\PP^3$, then by linear algebra we find
\beqv
I_2(C)=\linspan{x^2-xw,\,y^2-yw,\,x^2+y^2-z^2}.
\eeqv
Set-theoretically, there is no nonzero solution for $w=0$,
while for $w=1$ the first two relations force $x,y\in\{0,1\}$.
From this one easily derives that $Z_2(C)$ has dimension $0$ and $P_1,\dots,P_7$ are its only points.

However, by B\'ezout, $Z_2(C)$ should have length $8$.
It follows that $Z_2(C)$ is \emph{not} reduced.

Closer inspection shows that, indeed, $P_1$ has multiplicity $2$ in $Z_2(C)$ along the $z$-direction,
while its other points $P_2,\dots,P_7$ are simple.

We might write this informally as $Z_2(\{P_1,\dots,P_7\})=\{P_1+\vec{z},P_2,\dots,P_7\}$.
\end{example}

\entry
\label{dimC2}
The quadratic hull $Z_2(C)$ (or equivalently, its ideal $I_2(C)$) contains a lot of information about the square code $C\deux$.
First, by construction, we have
\beqv
\dim C\deux=\binom{k+1}{2}-\dim I_2(C)
\eeqv
where $k=\dim C$.

Further important properties will be given in section~\ref{further} at the end of this work.
These properties are not needed in the application to geometric realizations of multiplication reductions,
but they illustrate the significance of the quadratic hull.

\section{Application to geometric realizations}
\label{appl}

\entry\label{mainth}
We finish the discussion started in~\ref{almost}.
If $\phi:\cA\to\cB$ is a non-degenerate linear map
with projective system $\cP_\phi\in\PP\cA(\cB)$,
we set $I_2(\phi)=I_2(\cP_\phi)$.
Recall that $\phi$ induces an identification $\cA\simeq C_{\phi}$, $S^\cdot\cA\simeq S^\cdot C_{\phi}$, 
hence $\PP\cA\simeq\PP C_{\phi}$, under which $\cP_\phi$ corresponds to $\cP_{C_\phi}$.
It follows
\beqv
\vspace{.2\baselineskip}
\begin{split}
I_2(\phi)&=(S^2\phi)^{-1}I_2(C_\phi)\\
&=\ker(S^2\cA\overset{S^2\phi}{\simeq}S^2C_\phi\overset{m_{\cB}\,}{\longto}C_\phi)\\
&=\ker(S^2\cA\overset{S^2\phi\,}{\longto}S^2\cB\overset{m_{\cB}\,}{\longto}\cB).
\end{split}
\eeqv
We then define the quadratic hull of $\phi$ as
\beqv
Z_2(\phi)=Z_2(\cP_\phi)=Z(I_2(\phi))\subset\PP\cA,
\eeqv
which is the isomorphic image of $Z_2(C_\phi)\subset\PP C_{\phi}$ under the identification above.

\begin{theorem*}
Let $\phi:\cA\to\cB$ be a non-degenerate linear map.
Let $Q_{\cA}\in\PP\cA(\cA)$ be the canonical point of $\cA$,
and let $\cP_\phi\in\PP\cA(\cB)$ be the projective system of $\phi$,
and $Z_2(\phi)\subset\PP\cA$ its quadratic hull.

Then $\phi$ is a multiplication reduction $\cA\leadsto\cB$ if and only if $Z_2(\phi)$ contains $Q_{\cA}$.

Moreover, if this holds, then
\beqv
(Z_2(\phi),Q_{\cA},\cP_\phi,\cO(1)|_{Z_2(\phi)},V)
\eeqv
is a geometric realization of $\phi$,
where the linear system $V\subset\Gamma(Z_2(\phi),\cO(1)|_{Z_2(\phi)})$ is the isomorphic image of $\cA=\Gamma(\PP\cA,\cO(1))$
under the restriction-to-$Z_2(\phi)$ map.
\end{theorem*}
\begin{proof}
By elementary linear algebra, a necessary and sufficient condition for the existence of an $\omega:\cB\to\cA$ that completes the diagram
\beqv
\begin{CD}
S^2\cA  @>m_{\cA}>>\cA\\
@V{S^2\phi}VV @.\\
S^2\cB @>m_{\cB}>> \cB
\end{CD}
\eeqv
is that $\ker(S^2\cA\overset{S^2\phi\,}{\longto}S^2\cB\overset{m_{\cB}\,}{\longto}\cB)=I_2(\phi)$
is contained in $\ker(S^2\cA\overset{m_{\cA}\,}{\longto}\cA)=I_2(Q_{\cA})$.
This means precisely that $Q_{\cA}$ lies in $Z_2(\phi)$.

Now observe that the map $\cA=\Gamma(\PP\cA,\cO(1))\overset{\ev_{Q_{\cA}}}{\longto}\cA$ is the identity of $\cA$,
and factorizes through $V$ as
\beqv
\begin{CD}
\cA=\Gamma(\PP\cA,\cO(1)) @>{\res_{Z_2(\phi)}}>> V\subset\Gamma(Z_2(\phi),\cO(1)|_{Z_2(\phi)}) @>{\ev_{Q_{\cA}}}>> \cA,
\end{CD}
\eeqv
which shows that condition (i) in Proposition~\ref{evaluation-interpolation} is satisfied (moreover, with bijectivity).

Likewise the map $\cA=\Gamma(\PP\cA,\cO(1))\overset{\ev_{\cP_\phi}}{\longto}\cB$ is $\phi$,
and factorizes through $V$ as
\beqv
\begin{CD}
\cA=\Gamma(\PP\cA,\cO(1)) @>{\res_{Z_2(\phi)}}>> V\subset\Gamma(Z_2(\phi),\cO(1)|_{Z_2(\phi)}) @>{\ev_{\cP_\phi}}>> \cB,
\end{CD}
\eeqv
which shows that if condition (ii) is satisfied, then $(Z_2(\phi),Q_{\cA},\cP_\phi,\cO(1)|_{Z_2(\phi)},V)$
is a geometric realization for $\phi$.

In order to check (ii), we observe that the map $S^2\cA=\Gamma(\PP\cA,\cO(2))\overset{\ev_{\cP_\phi}}{\longto}\cB$
has kernel $I_2(\phi)$, and factorizes through $V\deux$ as
\beqv
\begin{CD}
S^2\cA=\Gamma(\PP\cA,\cO(2)) @>{\res_{Z_2(\phi)}}>> V\deux\subset\Gamma(Z_2(\phi),\cO(2)|_{Z_2(\phi)}) @>{\ev_{\cP_\phi}}>> \cB.
\end{CD}
\eeqv
But by definition $I_2(\phi)$ is also the kernel of $\Gamma(\PP\cA,\cO(2))\longto V\deux$.
Passing to the quotient we find that $V\deux\overset{\ev_{\cP_\phi}}{\longto}\cB$ is injective, as desired.
\end{proof}

\begin{remark}\label{explain}
The first part of the Theorem (together with its proof) is easily rephrased in elementary terms.

We are given a linear map $\phi:\cA\to\cB$, and we ask whether it is a multiplication algorithm, i.e. whether
it admits an adjoint $\omega:C_\phi\deux\to\cA$.

Choose a basis $a_1,\dots,a_k$ of $\cA$, and let $c_1=\phi(a_1)$, ..., $c_k=\phi(a_k)$ be the corresponding basis of $C$.

Any $y\in C_\phi$ is of the form $y=q(c_1,\dots,c_k)$, for $q\in\FF[x_1,\dots,x_k]_2$ a quadratic form in $k$ variables.
Then necessarily one should have $\omega(y)=q(a_1,\dots,a_k)$. As in \ref{adjoint} this shows the \emph{unicity} of $\omega$,
if it exists.

But now we're interested in the \emph{existence}. For this we observe that $y\in C_\phi$ can possibly be represented
by several quadratic forms. Then $\omega(y)$ will be well defined precisely when, for any such multiple
representation $y=q_1(c_1,\dots,c_k)=q_2(c_1,\dots,c_k)$, the corresponding values $q_1(a_1,\dots,a_k)=q_2(a_1,\dots,a_k)$
coincide. Setting $q=q_1-q_2$, this is equivalent to:

\emph{Any quadratic form $q\in\FF[x_1,\dots,x_k]_2$ vanishing at $(c_1,\dots,c_k)$ also vanishes at $(a_1,\dots,a_k)$.}

This exactly means that $Z_2(\phi)$ contains $Q_{\cA}$.
\end{remark}

\begin{corollary}\label{charsupercode}
Let $C$ be a $[\cB,k]$-code, where $k=\dim\cA$.
Then $C$ is the image of a supercode $\widehat{C}\subset\cA\times\cB$
under the second projection $\pi_{\cB}$
if and only if its quadratic hull $Z_2(C)\subset\PP C$
contains a point $Q\in Z_2(C)(\cA)$ with coordinates in $\cA$
that does not lie in a rational hyperplane of $\PP C$.
\end{corollary}
\begin{proof}
Choose a basis $c_1,\dots,c_k$ of $C$, and identify $\PP C\simeq\PP^{k-1}$ accordingly.
Assume there exists such a point $Q=(a_1:\dots:a_k)\in Z_2(C)(\cA)\subset\PP^{k-1}(\cA)$ not lying in a rational hyperplane.
Then $a_1,\dots,a_k$ are linearly independent, so they form a basis of $\cA$.
Identifying $\PP\cA\simeq\PP^{k-1}$ accordingly, we then have $Q=Q_{\cA}$.
Now if $\phi:\cA\to C\subset\cB$ is the linear map that sends $a_i$ to $c_i$,
then $\phi$ satisfies the conditions in the Theorem, so $\phi$ is a multiplication reduction.
Then $\widehat{C}=\im(\id_{\cA}\times\phi)$ is a supercode, and $C=C_\phi=\pi_{\cB}(\widehat{C})$ as wished.

The converse works similarly.
\end{proof}

\begin{remark}
The condition in the Corollary can be rephrased as the existence of a rational point on (a Zariski open set of) a certain scheme.
Indeed, points in $Z_2(C)(\cA)$ correspond to rational points of the Weil restriction of scalars of the base change to $\cA$ of $Z_2(C)$.
And the points not lying in a rational hyperplane of $\PP^{k-1}$ are those in the Zariski open defined by the non-vanishing of the determinant, in this Weil restriction.
\end{remark}

\begin{remark}
If $\FF=\F_q$ is a finite field, then there are only finitely many rational
hyperplanes in $\PP^{k-1}$, so their complement $U$ is a Zariski open set.
Then the condition in the Corollary also asks
\beqv
(Z_2(C)\cap U)(\cA)\neq\emptyset.
\eeqv
\end{remark}

\section{Experimental results}
\label{exp}

\entry
In this section we describe the quadratic hull of all the multiplication algorithms of minimal length
for some algebras of small cardinality over $\F_2$ and $\F_3$.
Let us first recall how these algorithms can be found exhaustively.

Let $\cA$ be an algebra over a field $\FF$. Given a basis $a_1,\dots,a_k$ of $\cA$, the product of
two elements $x,y\in\cA$ can be expanded as
\beqv
xy=B_1(x,y)a_1+\cdots+B_k(x,y)a_k,
\eeqv
which uniquely determines a family of $k$ symmetric bilinear forms $B_1,\dots,B_k$ on~$\cA$.
In order to fix notations, we state as a formal lemma what is only
a rephrasing of the definitions (and, as explained in the Introduction,
was already implicit in the very first papers on the topic):
\begin{lemma*}
A linear map $\phi=(\phi_1,\dots,\phi_n):\cA\to\FF^n$ is a multiplication algorithm if and only if
\beqv
B_1,\dots,B_k\in\linspan{\phi_1^{\tens 2},\dots,\phi_n^{\tens 2}}.
\eeqv
\end{lemma*}
Here $\linspan{\phi_1^{\tens 2},\dots,\phi_n^{\tens 2}}$ is the linear span of $\phi_1^{\tens 2},\dots,\phi_n^{\tens 2}$
inside $(\cA^\vee\tens\cA^\vee)^{\textrm{sym}}$, the space of symmetric tensors in $\cA^\vee\tens\cA^\vee$,
identified with the space of symmetric bilinear forms on $\cA$.
\begin{proof}[Proof of the Lemma]
Let $T_{\cA}\in\cA\tens(\cA^\vee\tens\cA^\vee)^{\textrm{sym}}$ be the multiplication tensor.
Then $\phi$ is a multiplication algorithm if and only if it has an adjoint~$\omega$,
i.e. if and only if one can write
\beqv
T_{\cA}=\omega_1\tens\phi_1^{\tens 2}+\cdots+\omega_n\tens\phi_n^{\tens 2}\;\in\;\cA\tens\linspan{\phi_1^{\tens 2},\dots,\phi_n^{\tens 2}}.
\eeqv
However, by definition of $B_1,\dots,B_k$, we also have
\beqv
T_{\cA}=a_1\tens B_1+\cdots+a_k\tens B_k.
\eeqv
Writing $\omega_1,\dots,\omega_n$ in the basis $a_1,\dots,a_k$ identifies these two expressions and concludes.
\end{proof}

Now in $V=(\cA^\vee\tens\cA^\vee)^{\textrm{sym}}$ set:
\begin{itemize}
\item $T=\linspan{B_1,\dots,B_k}$, or more intrinsically, $T=\{l\circ m_{\cA}:\:l\in\cA^\vee\}$
\item $\cE=\{l^{\tens 2}:\:l\in\cA^\vee\moins\{0\}\}$, the set of elementary symmetric tensors.
\end{itemize}

By the Lemma, if $\phi=(\phi_1,\dots,\phi_n)$ is a multiplication algorithm,
then $W=\linspan{\phi_1^{\tens 2},\dots,\phi_n^{\tens 2}}$ contains $T$;
moreover, if we assume that no shorter algorithm exists, then $\phi_1^{\tens 2},\dots,\phi_n^{\tens 2}$
are linearly independent, and $\dim W=n$.
Conversely, if $W\subset V$ is a subspace of dimension $n$ containing $T$
and admitting a basis $\phi_1^{\tens 2},\dots,\phi_n^{\tens 2}$ made of elements of $\cE$,
then  $\phi=(\phi_1,\dots,\phi_n)$ is a multiplication algorithm of length $n$ for $\cA$.

Thus, finding all these algorithms reduces to the following:

\begin{problem}
\label{pbVTE}
Let $V$ be a vector space over $\FF$, and $k,n$ two integers. Given:
\begin{itemize}
\item $T\subset V$ a ``target'' subspace of dimension $\dim T=k$
\item $\cE\subset V$ a set of generating elements.
\end{itemize}
Find all subspaces $W\subset V$ such that
\begin{enumerate}[(i)]
\item $T\subset W$
\item $W$ is generated by elements in $\cE$
\item $\dim W=n$.
\end{enumerate}
\end{problem}

A key observation of Oseledets \cite{Oseledets08} is that these $W$ are of the form $W=T\oplus\linspan{e_1,\dots,e_{n-k}}$
for some $e_1,\dots,e_{n-k}\in\cE$.
Thus they can be found by a search among $(n-k)$-tuples of elements of $\cE$,
instead of a naive search among $n$-tuples.

This was implemented systematically by Barbulescu, Detrey, Estibals and Zimmermann in \cite{BDEZ12},
which allowed them, among other results, to find all the optimal (symmetric) multiplication algorithms
for $\F_{2^k}$ over $\F_2$ for $k\leq 6$,
for $\F_{3^k}$ over $\F_3$ for $k\leq 5$,
for $\F_2[t]/(t^k)$ over $\F_2$ for $k\leq 6$,
for $\F_3[t]/(t^k)$ over $\F_3$ for $k\leq 5$,
for $\F_2[t]/(t^k-1)$ over $\F_2$ for $k\leq 7$,
and for $\F_3[t]/(t^k-1)$ over $\F_3$ for $k\leq 6$.
They wrote a computer program \cite{Zimmermann} that solves Problem~\ref{pbVTE} for arbitrary $V,T,\cE$ over $\F_2$ or $\F_3$.

Then in \cite{Covanov19} Covanov gives further improvements, using specifically
the geometry and symmetries of spaces of bilinear forms.

\entry
In the approach above, one space $W$ might admit several bases in $\cE$, hence correspond to several
multiplication algorithms.
However, an interesting fact is that all these algorithms will have the same quadratic hull, i.e. this quadratic hull only depends on $W$.

Indeed, first recall that there is a natural duality between $(\cA^\vee\tens\cA^\vee)^{\textrm{sym}}$
and $S^2\cA=\Gamma(\PP\cA,\cO(2))$.

In coordinates, if $S\in(\cA^\vee\tens\cA^\vee)^{\textrm{sym}}$
is written as a symmetric matrix
\beqv
S=(s_{i,j})_{\substack{1\leq i\leq k\\ 1\leq j\leq k}}
\eeqv
and if $Q\in\Gamma(\PP\cA,\cO(2))$ is written as a quadratic form
\beqv
Q=\sum_{1\leq i\leq j\leq k}q_{i,j}x_ix_j,
\eeqv
then this duality is given by
\beqv
\langle Q,S\rangle=\sum_{1\leq i\leq j\leq k}q_{i,j}s_{i,j}.
\eeqv
More intrinsically, if $S=l^{\tens 2}$ for $l\in\cA^\vee$, 
we have
\beqv
\langle Q,l^{\tens 2}\rangle=Q(l),
\eeqv
and then this is extended by linearity for general $S\in(\cA^\vee\tens\cA^\vee)^{\textrm{sym}}$.

\begin{proposition}
\label{I2phiperpW}
Let $\phi=(\phi_1,\dots,\phi_n):\cA\leadsto\FF^n$ be a multiplication algorithm,
let $W=\linspan{\phi_1^{\tens 2},\dots,\phi_n^{\tens 2}}$ be the linear span of $\phi_1^{\tens 2},\dots,\phi_n^{\tens 2}$
in $(\cA^\vee\tens\cA^\vee)^{\textrm{sym}}$, and let $I_2(\phi)\subset S^2\cA$ be the space of defining equations
for the quadratic hull $Z_2(\phi)$. Then, under the duality above, we have
\beqv
I_2(\phi)=W^\perp.
\eeqv
\end{proposition}
\begin{proof}
By definition, $I_2(\phi)=\{Q\in S^2\cA:\:Q(\phi_1)=\dots=Q(\phi_n)=0\}$. 
\end{proof}

\begin{corollary}
\label{pointsZ2phi}
The elementary tensors in $W$ are precisely the $l^{\tens 2}$ for $l\in Z_2(\phi)$ (up to a multiplicative constant).

As a consequence, we have $\#Z_2(\phi)(\F_q)\geq n$,
with equality if and only if $\phi$ is the only multiplication algorithm of length~$n$ corresponding to $W$.
\end{corollary}
\begin{proof}
Indeed, $l^{\tens 2}\in W=I_2(\phi)^\perp$ if and only if $Q(l)=0$ for all $Q\in I_2(\phi)$. 
\end{proof}

\begin{corollary}
\label{Z2phici}
Set $k=\dim\cA$, and let $\phi:\cA\leadsto\FF^n$ be a multiplication algorithm of minimal length,
i.e. with $n$ equal to the (symmetric) bilinear complexity $\mu_{\FF}^\mathrm{sym}(\cA)$.
Then $Z_2(\phi)$ has codimension at most $\binom{k+1}{2}-n$ in $\PP^{k-1}$,
with equality if and only if it is a complete intersection. 
\end{corollary}
\vspace{-\baselineskip}
\begin{proof}
Minimality of $n$ implies that $\phi_1^{\tens 2},\dots,\phi_n^{\tens 2}$ are linearly independent,
so $\dim W=n$, so $\dim I_2(\phi)=\binom{k+1}{2}-n$. Thus $Z_2(\phi)$ is defined by $\binom{k+1}{2}-n$
linearly independent quadratic equations, and we conlude.
\end{proof}

\vspace{-.5\baselineskip}
There are many examples of algebras $\cA$ for which one knows multiplication algorithms $\phi$ of
(possibly not minimal) length $n$ much smaller than $\binom{k+1}{2}$ --- for instance, $n$ linear in $k$.
This makes one expect $Z_2(\phi)$ to be low dimensional.

\entry
Zimmermann \cite{Zimmermann} kindly gave to the author a copy of the computer program that solves Problem~\ref{pbVTE}.
The output of this program is the list of the subspaces $W$, each given by a basis in row echelon form,
thus in general not made from elements of $\cE$.
By Proposition~\ref{I2phiperpW}, this suffices to find the quadratic hull of the corresponding multiplication
algorithms, with no need to compute these multiplication algorithms explicitly.

For $q=2$, $\cA_k=\F_{q^k}$, one gets the minimal length $n=\mu_q^\mathrm{sym}(\cA_k)$, the number of subspaces $W$,
and the quadratic hull of the corresponding algorithms, given as follows:

\vspace{.5\baselineskip}
\noindent\begin{minipage}{13cm}
\begin{center}
{\small Table 1: quadratic hull of minimal multiplication algorithms for $q=2$, $\cA_k=\F_{q^k}$}
\end{center}
\beqv
\begin{array}{|c|c|c|c|c|}
\hline
k & n & \textrm{nb. of $W$} & Z_2(\phi) & \#Z_2(\phi)(\F_q)\rule{0pt}{9pt}\\
\hline
2 & 3 & 1 & \PP^1 & 3\rule{0pt}{10pt}\\[1pt]
3 & 6 & 1 & \PP^2 & 7\\[1pt]
4 & 9 & 25 & \textrm{quadric}\simeq\PP^1\times\PP^1 & 9\\[1pt]
5 & 13 & 2015 & {\small\left\{\begin{array}{lc}310: & \textrm{smooth del Pezzo surface}\\[-2pt] 1705: & \textrm{singular del Pezzo surface}\end{array}\right.}& 13\\[3pt]
6 & 15 & 21 & {\small\begin{array}{c}\textrm{(connected) union of a pair}\\[-4pt] \textrm{of conjugate conics and five $\PP^1$}\end{array}} & 15\\[3pt]
\hline
\end{array}
\eeqv
\end{minipage}
\vspace{\baselineskip}

For $k\leq 5$, each of these quadratic hulls is a complete intersection.
In particular for $k=5$ (first case, smooth del Pezzo surface) one retrieves the example of section~\ref{dP}.

For $k=6$ we have $\codim Z_2(\phi)=4<\binom{k+1}{2}-n=6$, so we do not get a complete intersection.
Each of the two conjugate conics intersects each of the $\PP^1$ at one point. Otherwise no other pair of components intersect.
The canonical point is in one of the two conics.

\entry
For $q=3$, $\cA_k=\F_{q^k}$, we find likewise:

\vspace{.5\baselineskip}
\noindent\begin{minipage}{13cm}
\begin{center}
{\small Table 2: quadratic hull of minimal multiplication algorithms for $q=3$, $\cA_k=\F_{q^k}$}
\end{center}
\beqv
\begin{array}{|c|c|c|c|c|}
\hline
k & n & \textrm{nb. of $W$} & Z_2(\phi) & \#Z_2(\phi)(\F_q)\rule{0pt}{9pt}\\
\hline
2 & 3 & 1 & \PP^1 & 4\rule{0pt}{10pt}\\[1pt]
3 & 6 & 1 & \PP^2 & 13\\[1pt]
4 & 9 & 234 & {\small\left\{\begin{array}{lc}84: & \textrm{quadric}\simeq\mathbf{S}^2\\[-2pt] 150: & \textrm{quadric}\simeq\PP^1\times\PP^1\end{array}\right.} & {\small\begin{array}{c}10\\[-2pt] 16\end{array}}\\[2pt]
5 & 11 & 121 & {\small\begin{array}{c}\textrm{\{$11$ rational points and}\\[-2pt]\textrm{one point of degree $5$\}}\end{array}}& 11\\[1pt]
\hline
\end{array}
\eeqv
\end{minipage}
\vspace{\baselineskip}

These are all complete intersections.
In particular for $k=5$, $Z_2(\phi)$ is finite of length $16$, intersection of four quadrics;
actually it is the trivial geometric representation from~\ref{trivialrepr}.

\entry
For $q=2$, $\cA_k=\F_q[t]/(t^k)$, we get:

\vspace{.5\baselineskip}
\noindent\begin{minipage}{13cm}
\begin{center}
{\small Table 3: quadratic hull of minimal multiplication algorithms for $q=2$, $\cA_k=\F_q[t]/(t^k)$}
\end{center}
\beqv
\begin{array}{|c|c|c|c|c|}
\hline
k & n & \textrm{nb. of $W$} & Z_2(\phi) & \#Z_2(\phi)(\F_q)\rule{0pt}{9pt}\\
\hline
2 & 3 & 1 & \PP^1 & 3\rule{0pt}{10pt}\\[1pt]
3 & 5 & 2 & \PP^1\cup\PP^1 & 5\\[1pt]
4 & 8 & 4 & \PP^2\cup\PP^1 & 9\\[1pt]
5 & 11 & 112 & {\small\left\{\begin{array}{lc}96: & $\qquad$\textrm{quadric}\,\cup\PP^1\\[-2pt] 16: & \qquad\PP^2\cup\PP^2\cup\PP^1\end{array}\right.}& {\small\begin{array}{c}11\\[-2pt] 13\end{array}}\\[4pt]
6 & 14 & 384 & {\small\left\{\begin{array}{lc}256: & \textrm{quadric}\,\cup\PP^2\cup\PP^1\\[-2pt] 128: & \textrm{quadric}\,\cup\PP^1\cup\PP^1\cup\PP^1\end{array}\right.} & 15\\[2pt]
\hline
\end{array}
\eeqv
\end{minipage}
\vspace{\baselineskip}

These are complete intersections for $k=2$ or $3$, but not for $k=4,5,6$.

These are all connected.

For $k=3$, $k=4$, and the first case of $k=5$, the two components intersect precisely at a rational point, which is the support of the canonical point.

For the second case of $k=5$, the two $\PP^2$ intersect along a line, and then the $\PP^1$ intersects them at a point on this line, which is the support of the canonical point.

Likewise for the first case of $k=6$, the quadric and the $\PP^2$ intersect along a line, and then the $\PP^1$ intersects them at a point on this line, which is the support of the canonical point.

Last, for the second case of $k=6$, all components concur at one point, which is the support of the canonical point.

\entry
For $q=3$, $\cA_k=\F_q[t]/(t^k)$, we get:

\vspace{.5\baselineskip}
\noindent\begin{minipage}{13cm}
\begin{center}
{\small Table 4: quadratic hull of minimal multiplication algorithms for $q=3$, $\cA_k=\F_q[t]/(t^k)$}
\end{center}
\beqv
\begin{array}{|c|c|c|c|c|}
\hline
k & n & \textrm{nb. of $W$} & Z_2(\phi) & \#Z_2(\phi)(\F_q)\rule{0pt}{9pt}\\
\hline
2 & 3 & 1 & \PP^1 & 4\rule{0pt}{10pt}\\[1pt]
3 & 5 & 3 & \PP^1\cup\PP^1 & 7\\[1pt]
4 & 8 & 252 & {\small\left\{\begin{array}{lc}243: & \quad\textrm{conic}\,\cup\PP^1\cup\PP^1\\[-2pt] 9: & \qquad\PP^2\cup\PP^1\end{array}\right.}  & {\small\begin{array}{c}10\\[-2pt] 16\end{array}}\\[6pt]
5 & 10 & 243 & \textrm{conic}\,\cup\PP^1\cup\PP^1\cup\PP^1 & 13\\[1pt]
\hline
\end{array}
\eeqv
\end{minipage}
\vspace{\baselineskip}

These are complete intersections for $k=2$ or $3$, but not for $k=4$ or $5$.

These are all connected.
For $k=3,4,5$, the components all concur at one point, which is the support of the canonical point.

\entry
Minimal algorithms were also computed for $\cA_k=\F_q[t]/(t^k-1)$, with $q=2$ and $k\leq 7$, and with $q=3$ and $k\leq 6$.
In general $\F_q[t]/(t^k-1)$ is not a local algebra, so it decomposes as the product of its local factors.
It turns out all these examples follow Strassen's (strong) direct sum conjecture: one finds that
the minimal algorithms for $\cA_k$ are obtained by taking products of minimal algorithms for the local factors.
And then, the quadratic hull of the product algorithm is the disjoint union of the quadratic hulls of the factors. 

For instance, we have $\F_2[t]/(t^7-1)\simeq\F_2\times\F_{2^3}\times\F_{2^3}$,
and as predicted by the conjecture, one finds $49$ multiplication algorithms
for $\F_2[t]/(t^7-1)$ of minimal length $n=13$, all of which decomposing
as a product of a minimal algorithm for $\F_2$ ($1$ choice, length $1$),
and a minimal algorithm for each copy of $\F_{2^3}$ ($7$ choices each, length $6$).
All these $49$ multiplication algorithms come from the same subspace $W$,
hence have the same quadratic hull, which is a disjoint union of a point and two copies of $\PP^2$. 

\section{Further properties of the quadratic hull}
\label{further}

\entry\label{extensions}
Given a $[\cB,k]$-code $C\subset\cB$, and a point $P\in\PP C(\FF)$
representing a nonzero linear form $\lambda_P:C\surj\FF$ (up to a multiplicative constant),
the extended code $C_P\subset\cB\times\FF$ is the $[\cB\times\FF,k]$-code
made of the words of the form
\beqv
(c,\lambda_P(c))
\eeqv
for $c\in C$.
(In particular, if $\cB=\FF^n$ and $C$ is a $[n,k]$-code, then $C_P$ is the extended $[n+1,k]$-code of $C$ by $\lambda_P$
in the usual sense.)

In coordinates, if $\GG$ is a generator matrix for $C$,
and if $P\in\PP^{k-1}(\FF)$ is represented as a column vector in $\FF^k$,
then $C_P$ is the code with generator matrix $\GG_P=(\;\GG\;|P)$.

\begin{proposition*}
With these notations, we have either $\dim C_P\deux=\dim C\deux$ or $\dim C_P\deux=\dim C\deux+1$.

Moreover, the rational points of $Z_2(C)$ parameterize \emph{square-preserving} extensions $C_P$ of $C$, i.e. those such that
\beqv
\dim C_P\deux=\dim C\deux.
\eeqv
\end{proposition*}
\begin{proof}
Take $P\in\PP^{k-1}(\FF)$, write $C\deux$ as the image of the evaluation map
\beqv
\begin{array}{ccc}
\Gamma(\PP^{k-1},\cO(2)) & \longto & \cB\\
q & \mapsto & (q(P_1),\dots,q(P_m))
\end{array}
\eeqv
whose kernel is $I_2(C)=I_2(\{P_1,\dots,P_m\})$,
and write $C_P\deux$ as the image of the evaluation map
\beqv
\begin{array}{ccc}
\Gamma(\PP^{k-1},\cO(2)) & \longto & \cB\times\FF\\
q & \mapsto & (q(P_1),\dots,q(P_m),q(P))
\end{array}
\eeqv
whose kernel is $I_2(C_P)=I_2(\{P_1,\dots,P_m,P\})=I_2(C)\cap I_2(P)$.
However, $I_2(P)$ is a hyperplane in $\Gamma(\PP^{k-1},\cO(2)$, so
\begin{itemize}
\item either $I_2(C)\not\subset I_2(P)$, and then $I_2(C_P)\subsetneq I_2(C)$ has codimension~$1$, and $\dim C_P\deux=\dim C\deux+1$
\item or $I_2(C)\subset I_2(P)$, thus $I_2(C_P)=I_2(C)$ and $\dim C_P\deux=\dim C\deux$.
\end{itemize}
The latter means precisely that $P$ lies in $Z_2(C)$.
\end{proof}

\entry
Recall that points of $\PP^{k-1}\simeq\PP C$ also parameterize hyperplanes of $C$.

If $P\in\PP^{k-1}(\FF)$, its orthogonal $P^\perp$ is a hyperplane in $\FF^k$. The image of $P^\perp$ under the generator matrix $\GG$
is then the hyperplane $H\subset C$ corresponding to $P$.

Equivalently, linear forms vanishing at $P$ form a hyperplane $\Gamma(\PP^{k-1},\fI_P\cO(1))\subset\Gamma(\PP^{k-1},\cO(1))$.
Then we have
\beqv
H=\ev_{\cP_C}(\Gamma(\PP^{k-1},\fI_P\cO(1))).
\eeqv

\begin{proposition*}\label{HC}
Rational points of $Z_2(C)$ parameterize hyperplanes $H\subset C$ such that
\beqv
H\cdot C\subsetneq C\deux.
\eeqv
\end{proposition*}
\begin{proof}
Evaluation at $\cP_C=(P_1,\dots,P_m)$ identifies $\Gamma(\PP^{k-1},\cO(1))$ with $C$,
and maps $\Gamma(\PP^{k-1},\cO(2))$ onto $C\deux$, compatibly with multiplication.

Writing $H=\ev_{\cP_C}(\Gamma(\PP^{k-1},\fI_P\cO(1)))$ we find
\beqv
\begin{split}
H\cdot C&=\ev_{\cP_C}(\Gamma(\PP^{k-1},\fI_P\cO(1)))\cdot\ev_{\cP_C}(\Gamma(\PP^{k-1},\cO(1)))\\
&=\ev_{\cP_C}(\Gamma(\PP^{k-1},\fI_P\cO(1))\cdot\Gamma(\PP^{k-1},\cO(1)))\\
&=\ev_{\cP_C}(\Gamma(\PP^{k-1},\fI_P\cO(2)))
\end{split}
\eeqv
where
\beqv
\Gamma(\PP^{k-1},\fI_P\cO(2))=I_2(P)
\eeqv
is the hyperplane in $\Gamma(\PP^{k-1},\cO(2))$ of quadratic forms vanishing at $P$. 

We thus have a commutative diagram
\beqv
\begin{CD}
0 @>>> I_2(\{P_1,\dots,P_m\}) @>>> \Gamma(\PP^{k-1},\cO(2)) @>\ev_{\cP_C}>> C\deux @>>>0\\
@. @AAA @AAA @AAA @.\\
0 @>>> I_2(\{P_1,\dots,P_m,P\}) @>>> I_2(P) @>\ev_{\cP_C}>> H\cdot C @>>>0
\end{CD}
\eeqv
with exact rows, and injective vertical maps.
Since $I_2(P)$ has codimension $1$ in $\Gamma(\PP^{k-1},\cO(2))$, then
\begin{itemize}[leftmargin=\parindent]
\item either $H\cdot C=C\deux$ and $I_2(\{P_1,\dots,P_m,P\})\subsetneq I_2(\{P_1,\dots,P_m\})$ has codimension $1$
\item or $H\cdot C\subsetneq C\deux$ has codimension $1$ and $I_2(\{P_1,\dots,P_m,P\})=I_2(\{P_1,\dots,P_m\})$.
\end{itemize}
The latter occurs precisely when $P$ lies in $Z_2(\{P_1,\dots,P_m\})=Z_2(C)$.
\vspace{.5\baselineskip}
\end{proof}

\entry\label{H2secante}
Recall that set-theoretically, the secant variety $\Sec(Z)$ of a closed subscheme $Z$ in a projective space,
is the union of all lines, possibly defined over an extension field,
whose intersection with $Z$ has length at least $2$.
This intersection could consist of non-rational points, or of points with multiplicities, so $\Sec(Z)$ really depends
on the scheme structure of $Z$.

\begin{proposition*}
Rational points of the secant variety $\Sec(Z_2(C))$ parameterize a subset of
(and possibly all) the locus of hyperplanes $H\subset C$  that satisfy
\beqv
H\deux\subsetneq C\deux.
\eeqv
\end{proposition*}
\begin{proof}
Let $P\in\Sec(Z_2(C))(\FF)$ parameterize a hyperplane $H\subsetneq C$, and assume by contradiction $H\deux=C\deux$.

We have $H=\ev_{\cP_C}(\Gamma(\PP^{k-1},\fI_P\cO(1)))$, so
\beqv
H\deux=\ev_{\cP_C}(\Gamma(\PP^{k-1},\fI_P^2\cO(2)))
\eeqv
is the evaluation of the space of quadratic forms vanishing at $P$ with multiplicity~$2$.

From the diagram
\beqv
\begin{CD}
0 @>>> I_2(\{P_1,\dots,P_m\}) @>>> \Gamma(\PP^{k-1},\cO(2)) @>\ev_{\cP_C}>> C\deux @>>>0\\
@. @. @AAA @| @.\\
@. @. \Gamma(\PP^{k-1},\fI_P^2\cO(2))  @>\ev_{\cP_C}>> H\deux @>>>0
\end{CD}
\eeqv
we deduce that
\begin{itemize}
\item $\Gamma(\PP^{k-1},\fI_P^2\cO(2))$ maps onto $\Gamma(\PP^{k-1},\cO(2))/I_2(\{P_1,\dots,P_m\})$
\item thus, $\Gamma(\PP^{k-1},\fI_P^2\cO(2))+I_2(\{P_1,\dots,P_m\})=\Gamma(\PP^{k-1},\cO(2))$
\item thus, $I_2(\{P_1,\dots,P_m\})$ maps onto $\Gamma(\PP^{k-1},\cO(2))/\Gamma(\PP^{k-1},\fI_P^2\cO(2))$
\end{itemize}
(actually these three statements are equivalent).

After possibly extension of scalars, we restrict the last assertion to a line $L$ passing through $P$
and having intersection with $Z_2(C)$ of length at least $2$.
It gives that $I_2(\{P_1,\dots,P_m\})|_L$ maps onto $\Gamma(L,\cO(2))/\Gamma(L,\fI_P^2\cO(2))$,
which has dimension~$2$.

However, since $I_2(\{P_1,\dots,P_m\})$ is a space of defining equations for $Z_2(C)$,
its restriction $I_2(\{P_1,\dots,P_m\})|_L$ lives in $\Gamma(L,\fI_{Z_2(C)\cap L}\cO(2))$.
But $Z_2(C)\cap L$ has length at least $2$, so $\Gamma(L,\fI_{Z_2(C)\cap L}\cO(2))$ has dimension at most $1$, a contradiction.
\end{proof}

\entry\label{courbe}
In \cite[Th.~2]{MMP14}, M\'arquez-Corbella, Mart\'inez-Moro, and Pellikaan
showed that if $C$ is the evaluation code of some $L(D)$ at $n$ distinct rational points on a smooth projective curve $X$ of genus~$g$,
with $2g+1<\deg(D)<n/2$, then $Z_2(C)=X$.
We slightly strengthen this result, and also generalize it to $\cB$-codes:

\begin{proposition*}
Let $X$ be a smooth projective curve of genus~$g$, let $D$ be a divisor on $X$, and $\cP\in X(\cB)$.
Choose a trivialization $\cO(D)|_{\cP}\simeq\cB$, and let
\beqv
C=\im(L(D)\overset{\res_{\cP}}{\longto}\cO(D)|_{\cP}\simeq\cB)
\eeqv
be the evaluation $\cB$-code of $L(D)$ at $\cP$.
Let $\im(\cP)$ be the schematic image of $\cP$ in $X$, which is a finite closed subscheme of $X$,
and let $[\im(\cP)]$ be its fundamental cycle \cite[\S1.5]{Fulton}, seen as a divisor on $X$. Assume
\beqv
\deg(D)>2g+1
\eeqv
and
\beqv
L(2D-[\im(\cP)])=0.
\eeqv
Then
\beqv
Z_2(C)=X.
\eeqv
\end{proposition*}
\begin{proof}
Set $k=\dim(C)=\dim L(D)$, and identify $X$ with its image under
the embedding $\iota:X\inj\PP^{k-1}$ defined by $L(D)$.
Observe that we then have $\cP_C=\iota\circ\cP$.

By \cite{SaintDonat72} and the hypothesis $\deg(D)>2g+1$, $\;X\subset\PP^{k-1}$ is defined by quadratic equations, and it contains $\cP_C$,
thus
\beqv
Z_2(C)\subset X.
\eeqv

Conversely, let $f\in I_2(C)\subset\Gamma(\PP^{k-1},\cO(2))$.
Then $f|_X$ lives in $L(2D)$ and vanishes at the image of $\cP$, so actually $f|_X\in L(2D-[\im(\cP)])$,
which is $0$ by hypothesis.
This forces $f|_X=0$, which means $f\in I_2(X)$. This proves $I_2(C)\subset I_2(X)$, hence
\beqv
Z_2(C)\supset X.
\eeqv
\end{proof}

If $\cB=\F^n$ and $\cP$ is a collection of \emph{distinct} rational points $P_1,\dots,P_n$,
then $[\im(\cP)]=P_1+\cdots+P_n$.
If moreover $\deg(D)<n/2$ then $L(2D-[\im(\cP)])=0$
for degree reasons, so Proposition~\ref{courbe} implies \cite[Th.~2]{MMP14}.
Is this slight improvement useful?

Let $q\geq49$ be a square. As $k\to\infty$, the best multiplication algorithms
for $\F_{q^k}$ over $\F_q$ in the literature \cite[Th.~6.4]{HR-ChCh+} use evaluation-interpolation on a curve $X$
at $n=2k+g-1$ points, for a divisor $D$ of degree $\deg(D)=k+g-1$
(actually this is also the case in older
results such as \cite[Th.~1.1]{Ballet99}, that require curves
with a larger number of points but only use $2k+g-1$ of them),
with $g\approx\frac{2k}{\sqrt{q}-2}$.
As a consequence, condition $\deg(D)>2g+1$ is satisfied, 
but condition $\deg(D)<n/2$ is \emph{not}.
So \cite[Th.~2]{MMP14} does not apply to the associated code $C$.
However, the very construction of these algorithms ensures $L(2D-(P_1+\cdots+P_n))=0$,
so the conclusion $Z_2(C)=X$ still holds thanks to Proposition~\ref{courbe}.

\end{document}